\documentclass[a4paper,UKenglish,cleveref, autoref]{lipics-v2019}
\usepackage[utf8]{inputenc}
\usepackage{amsmath,amssymb}
\usepackage{hyperref}
\usepackage{amsthm}
\usepackage{tikz}
\usepackage{slashbox}
\usepackage{boxedminipage}
\usepackage{xypic}

\newcommand{\W}{\ensuremath{\mathrm{W}}}
\newcommand{\NP}{{\sf NP}}
\newcommand{\ssi}{\subseteq_i}
\newcommand{\ssin}{\subsetneq_i}
\newcommand{\si}{\supseteq_i}

\newtheorem{open}{Open Problem}

\title{Acyclic, Star and Injective Colouring: A~Complexity Picture for $H$-Free Graphs\footnote{Extended abstracts of this paper are in the proceedings of ESA~2020~\cite{BJMPS20} and CSR~2021~\cite{BJMPS21}.}}

\titlerunning{Acyclic Colouring, Star Colouring and Injective Colouring for ${\mathbf H}$-Free Graphs}

\author{Jan Bok}{Computer Science Institute, 
Charles University, Prague, Czech Republic}{bok@iuuk.mff.cuni.cz}{https://orcid.org/0000-0002-7973-1361}{Supported by GAUK 1580119, SVV–2020–260578, and by The European Research Council (ERC) under the European Union’s Horizon 2020 research and
innovation programme (grant agreement No 810115 - DYNASNET).}
\author{Nikola Jedli\u{c}kov\'{a}}{Department of Applied Mathematics,
Charles University, Prague, Czech Republic}{jedlickova@kam.mff.cuni.cz}{https://orcid.org/0000-0001-9518-6386}{Supported by GAUK 1198419 and SVV--2020--260578.}
\author{Barnaby Martin}{Department of Computer Science, Durham University, Durham, United Kingdom}{barnaby.d.martin@durham.ac.uk}{}{}
\author{Pascal Ochem}{LIRMM, CNRS, Universit\'e de Montpellier, Montpellier, France}{ochem@lirmm.fr}{}{}
\author{Dani\"el Paulusma}{Department of Computer Science, Durham University, Durham United Kingdom}{daniel.paulusma@durham.ac.uk}{0000-0001-5945-9287}{Supported by the Leverhulme Trust (RPG-2016-258).}
\author{Siani Smith}{Department of Computer Science, Durham University, Durham, United Kingdom}{siani.smith@durham.ac.uk}{}{}

\authorrunning{J. Bok, N. Jedli{\u{c}}kov\'{a}, B. Martin, P. Ochem, D. Paulusma and S. Smith}
\Copyright{Jan Bok, Nikola Jedli{}\u{c}kov\'{a}, Barnaby Martin, Pascal Ochem, Dani\"el Paulusma and Siani Smith}



\nolinenumbers
\hideLIPIcs

\newcommand{\problemdef}[3]{
	\begin{center}
		\begin{boxedminipage}{.99\textwidth}
			\textsc{{#1}}\\[2pt]
			\begin{tabular}{ r p{0.8\textwidth}}
				\textit{~~~~Instance:} & {#2}\\
				\textit{Question:} & {#3}
			\end{tabular}
		\end{boxedminipage}
	\end{center}
}

\begin{document}

\maketitle

\begin{abstract}
A (proper) colouring is acyclic, star, or injective if any two colour classes induce a forest, star forest or disjoint union of vertices and edges, respectively.
Hence, every injective colouring is a star colouring and every star colouring is an acyclic colouring.
The corresponding decision problems are \textsc{Acyclic Colouring}, \textsc{Star Colouring} and \textsc{Injective Colouring} (the last problem is also known as \textsc{$L(1,1)$-Labelling}).
A classical complexity result on \textsc{Colouring} is a well-known dichotomy for $H$-free graphs (a graph is $H$-free if it does not contain $H$ as an \emph{induced} subgraph).
In contrast, there is no systematic study into the computational complexity of \textsc{Acyclic Colouring}, \textsc{Star Colouring} and \textsc{Injective Colouring} despite numerous algorithmic and structural results that have appeared over the years.
We perform such a study and give almost complete complexity classifications for
\textsc{Acyclic Colouring}, \textsc{Star Colouring} and \textsc{Injective Colouring} on $H$-free graphs
(for each of the problems, we have one open case).
Moreover, we give full complexity classifications if the number of colours $k$ is fixed, that is, not part of the input.
From our study it follows that for fixed~$k$ the three problems behave in the same way, but this is no longer true if $k$ is part of the input.
To obtain several of our results we prove stronger complexity results that in particular involve the girth of a graph and the class of line graphs of multigraphs.
\end{abstract}

\section{Introduction}\label{s-intro}

We study the complexity of three classical colouring problems. We do this by focusing on \emph{hereditary} graph classes, i.e., classes closed under vertex deletion, or equivalently, classes characterized by a (possibly infinite) set ${\cal F}$ of forbidden induced subgraphs. As evidenced by numerous complexity studies in the literature, even the case where $|{\cal F}|=1$ captures a rich family of graph classes suitably interesting to develop general methodology. Hence, we usually first set ${\cal F}=\{H\}$ and consider the class of \emph{$H$-free} graphs, i.e., graphs that do not contain $H$ as an induced subgraph. We then investigate how the complexity of a problem restricted to $H$-free graphs depends on the choice of $H$ and try to obtain a \emph{complexity dichotomy}.

To give a well-known and relevant example, the \textsc{Colouring} problem is to decide, given a graph $G$ and integer $k\geq 1$, if $G$
has a \emph{$k$-colouring}, i.e., a mapping $c:V(G)\to \{1,\ldots,k\}$ such that $c(u)\neq c(v)$ for every two adjacent vertices $u$ and $v$.
Kr\'al' et al.~\cite{KKTW01} proved that \textsc{Colouring} on $H$-free graphs is polynomial-time solvable if $H$ is an induced subgraph of $P_4$ or $P_1+P_3$ and \NP-complete otherwise. Here, $P_n$ denotes the $n$-vertex path and $G_1+G_2=(V(G_1)\cup V(G_2),E(G_1)\cup E(G_2))$ the disjoint union of two vertex-disjoint graphs $G_1$ and $G_2$.
If $k$ is fixed (not part of the input), then we obtain the \textsc{$k$-Colouring} problem. No complexity dichotomy is known for $k$-\textsc{Colouring} if $k\geq 3$.
In particular, the complexities of $3$-\textsc{Colouring} for $P_t$-free graphs for $t\geq 8$ and $k$-\textsc{Colouring} for $sP_3$-free graphs for $s\geq 2$ and $k\geq 4$ are still open. Here, we write $sG$ for the disjoint union of $s$ copies of $G$. We refer to the survey of Golovach et al.~\cite{GJPS17} for further details and to~\cite{CHSZ18,KMMNPS18} for updated summaries.

For a colouring $c$ of a graph $G$, a \emph{colour class} consists of all vertices of $G$ that are mapped by $c$ to a specific colour~$i$.
We consider the following special graph colourings. A colouring of a graph $G$ is \emph{acyclic} if the union of any two colour classes induces a forest.
The $(r+1)$-vertex \emph{star} $K_{1,r}$ is the graph with vertices $u,v_1,\ldots, v_r$ and edges $uv_i$ for every $i\in \{1,\ldots,r\}$.
An acyclic colouring is a \emph{star colouring} if the union of any two colour classes induces a \emph{star forest}, that is, a forest in which each connected component is a star.
A star colouring is \emph{injective} (or an \emph{$L(1,1)$-labelling} or a \emph{distance-$2$ colouring}) if the union of any two colour classes induces
an $sP_1+tP_2$ for some integers $s\geq 0$ and $t\geq 0$. An alternative definition is to say that all the neighbours of every vertex of $G$ are uniquely coloured.
These definitions lead to the following three decision problems:

\problemdef{Acyclic Colouring}{A graph $G$ and an integer $k\geq 1$}{Does $G$ have an acyclic $k$-colouring?}

\problemdef{Star Colouring}{A graph $G$ and an integer $k\geq 1$}{Does $G$ have a star $k$-colouring?}

\problemdef{Injective Colouring}{A graph $G$ and an integer $k\geq 1$}{Does $G$ have an injective $k$-colouring?}

\noindent
If $k$ is fixed, we write \textsc{Acyclic $k$-Colouring}, \textsc{Star $k$-Colouring} and \textsc{Injective $k$-Colouring}, respectively.

All three problems have been extensively studied. We note that in the literature on the \textsc{Injective Colouring} problem it is often assumed that two adjacent vertices may be coloured alike by an injective colouring (see, for example,~\cite{HKSS02,HRS08,JXZ13}).
However, in our paper, we do {\bf not} allow this; as reflected in their definitions we only consider colourings that are proper. This enables us to compare the results for the three different kinds of colourings with each other.

So far, systematic studies mainly focused on structural characterizations, exact values, lower and upper bounds on the minimum number of colours in an acyclic colouring or star colouring
(i.e., the \emph{acyclic} and \emph{star chromatic number}); 
see, e.g.,~\cite{AMR91,Bo79,FGR02,FR08,FRR04,KM11,Ka18,Wo05,ZB04,ZLSX16}, to name just a few papers,
whereas injective colourings (and the \emph{injective chromatic number}) were mainly considered in the context of the distance constrained labelling framework (see the survey~\cite{Ca11} and Section~\ref{s-con} therein).
The problems have also been studied from a complexity perspective, but apart from a study on \textsc{Acyclic Colouring} for graphs of bounded maximum degree~\cite{MNRW13},
known results are scattered and relatively sparse. We perform a \emph{systematic} and \emph{comparative} complexity study by focusing on the following research question both for $k$ part of the input and for fixed $k$:

\medskip
\noindent
\emph{What are the computational complexities of {\sc Acyclic Colouring}, {\sc Star Colouring} and {\sc Injective Colouring} for $H$-free graphs?}

\paragraph*{Known Results}

Before discussing our new results and techniques, we first briefly discuss some known results.

Coleman and Cai~\cite{CC86} proved that for every $k\geq 3$, \textsc{Acyclic $k$-Colouring} is \NP-complete for bipartite graphs.
Afterwards, a number of hardness results appeared for other hereditary graph classes.
Alon and Zaks~\cite{AZ02} showed that \textsc{Acyclic $3$-Colouring} is \NP-complete for line graphs of maximum degree~$4$.
Kostochka~\cite{Ko78} proved that \textsc{Acyclic $3$-Colouring} is \NP-complete for planar graphs.
This result was improved to planar bipartite graphs of maximum degree~$4$ by Ochem~\cite{Oc05}.
Mondal et al.~\cite{MNRW13} proved that \textsc{Acyclic $4$-Colouring} is \NP-complete for graphs of maximum degree~$5$ and for planar graphs of maximum degree~$7$.
Ochem~\cite{Oc05} showed that \textsc{Acyclic $4$-Colouring} is \NP-complete for planar bipartite graphs of maximum degree~$8$.
We refer to the paper of Angelini and Frati~\cite{AF12} for a further discussion on acyclic colourable planar graphs.

Albertson et al.~\cite{ACKKR04} and recently, Lei et al.~\cite{LSS18} proved that \textsc{Star $3$-Colouring} is \NP-complete for planar bipartite graphs and line graphs, respectively.
Shalu and Antony~\cite{SA20} showed that \textsc{Star Colouring} is \NP-complete for co-bipartite graphs.
Bodlaender et al.~\cite{BKTL04}, Sen and Huson~\cite{SH97} and Lloyd and Ramanathan~\cite{LR92}
proved that \textsc{Injective Colouring} is \NP-complete for split graphs, unit disk graphs and planar graphs, respectively.
Mahdian~\cite{Ma02} proved that for every $k\geq 4$, \textsc{Injective $k$-Colouring} is \NP-complete for line graphs,
whereas \textsc{Injective $4$-Colouring} is also known to be \NP-complete for cubic graphs (see~\cite{Ca11}). Observe that \textsc{Injective $3$-Colouring} is trivial for general graphs.

On the positive side, Lyons~\cite{Ly11} proved that \textsc{Acyclic Colouring} and \textsc{Star Colouring} are polynomial-time solvable for $P_4$-free graphs; in particular, he showed that every acyclic colouring of a $P_4$-free graph is, in fact, a star colouring.
We note that \textsc{Injective Colouring} is trivial for $P_4$-free graphs, as every injective colouring must assign each vertex of a connected $P_4$-free graph a unique colour.
Afterwards, the results of Lyons have been extended to $P_4$-tidy graphs and ($q,q-4)$-graphs by Linhares-Sales et al.~\cite{SMMS14}.

Cheng et al.~\cite{CMS11} complemented the aforementioned result of Alon and Zaks~\cite{AZ02} by proving that \textsc{Acyclic Colouring} is polynomial-time solvable for claw-free graphs of maximum degree at most~$3$.
Calamoneri~\cite{Ca11} observed that \textsc{Injective Colouring} is polynomial-time solvable for comparability and co-comparability graphs.
Zhou et al.~\cite{ZKN00} proved that \textsc{Injective Colouring} is polynomial-time solvable for graphs of bounded treewidth (which is best possible due to the \W[1]-hardness result of Fiala et al.~\cite{FGK11}).

Finally, we refer to~\cite{BGMPS} for a complexity study of \textsc{Acyclic Colouring}, \textsc{Star Colouring} and \textsc{Injective Colouring} for graphs of bounded diameter.

\paragraph*{Our Complexity Results and Methodology}
The \emph{girth} of a graph~$G$ is the length of a shortest cycle of $G$ (if $G$ is a forest, then its girth is $\infty$).
To answer our research question we focus on two important graph classes, namely the classes of graphs of high girth and line graphs of multigraphs, which are interesting classes on their own.
If a problem is \NP-complete for both classes, then it is \NP-complete for $H$-free graphs whenever $H$ has a cycle or a claw. It then remains to analyze the case when $H$ is a \emph{linear forest}, i.e., a disjoint union of paths; see~\cite{BDFJP19,BGPS12,GLPR19,KKTW01} for examples of this approach, which we discuss in detail below.

The construction of graph families of high girth and large chromatic number is well studied in graph theory (see, e.g.~\cite{Er59}).
To prove their complexity dichotomy for \textsc{Colouring} on $H$-free graphs, Kr\'al' et al.~\cite{KKTW01} first showed that for every integer $g\geq 3$, \textsc{$3$-Colouring} is \NP-complete for the class of graphs of girth at least~$g$. This approach can be readily extended to any integer $k\geq 3$~\cite{EHK98,LK07}. The basic idea is to replace edges in a graph by graphs of high girth and large chromatic number,
such that the resulting graph has sufficiently high girth and is $k$-colourable if and only if the original graph is so (see also~\cite{GPS14,HJP15}).
 
By a more intricate use of the above technique we are able to prove that for every $g\geq 3$ and every $k\geq 3$, \textsc{Acyclic $k$-Colouring} is \NP-complete for the class of $2$-degenerate bipartite graphs of girth at least~$g$.
This implies that \textsc{Acyclic $k$-Colouring} is \NP-complete for $H$-free graphs whenever $H$ has a cycle.
We are also able to prove that for every $g\geq 3$, \textsc{Star $3$-Colouring} remains \NP-complete even for planar graphs of girth at least~$g$ and maximum degree~$3$.
This implies that \textsc{Star $3$-Colouring} is \NP-complete for $H$-free graphs whenever $H$ has a cycle.
We prove the latter result for every $k\geq 4$ by combining known results, just as we use known results to prove that \textsc{Injective $k$-Colouring} $(k\geq 4$) is \NP-complete for $H$-free graphs if $H$ has a cycle.

A classical result of Holyer~\cite{Ho81} is that \textsc{$3$-Colouring} is \NP-complete for line graphs (and Leven and Galil~\cite{LG83} proved the same for $k\geq 4$).
As line graphs are claw-free, Kr\'al' et al.~\cite{KKTW01} used Holyer's result to show that \textsc{$3$-Colouring} is \NP-complete for $H$-free graphs whenever $H$ has an induced claw.
For \textsc{Acyclic $k$-Colouring}, we can use Alon and Zaks' result~\cite{AZ02} for $k=3$, which we extend to work for $k\geq 4$.
For \textsc{Star $k$-Colouring} we extend the recent result of Lei et al.~\cite{LSS18} from $k=3$ to $k\geq 3$ (in both our results we consider line graphs of multigraphs; these graphs are claw-free and hence suffice for our study on $H$-free graphs).
For \textsc{Injective $k$-Colouring} $(k\geq 4$) we can use the aforementioned result on line graphs of Mahdian~\cite{Ma02}.

The above hardness results leave us to consider the case where $H$ is a linear forest. In Section~\ref{s-gen} we will use a result of Atminas et al.~\cite{ALR12} to prove a general result from which it follows that for fixed $k$, all three problems are polynomial-time solvable for $H$-free graphs if $H$ is a linear forest. Hence, we have full complexity dichotomies for the three problems when $k$ is fixed.
However, these positive results do not extend to the case where $k$ is part of the input.
That is, for each of the three problems, we prove \NP-completeness for graphs that are $P_r$-free for some small value of $r$ or have a small independence number, i.e., that are $sP_1$-free for some small integer $s$.

Our complexity results for $H$-free graphs are summarized in the following three theorems, proven in Sections~\ref{s-acyclic}--\ref{s-injective}, respectively; see Table~\ref{t-thetable} for a comparison. For two graphs $F$ and $G$, we write $F\ssi G$ or $G\si F$ to denote that $F$ is an \emph{induced} subgraph of $G$.

\begin{theorem}\label{t-acyclic}
Let $H$ be a graph. For the class of $H$-free graphs it holds that:\\[-10pt]
\begin{enumerate}
\item [(i)] {\sc Acyclic Colouring} is polynomial-time solvable if $H\ssi P_4$ and \NP-complete if $H\not\ssi P_4$ and $H\neq 2P_2$;\\[-10pt]
\item [(ii)] for every $k\geq 3$, {\sc Acyclic $k$-Colouring} is polynomial-time solvable if $H$ is a linear forest and \NP-complete otherwise.
\end{enumerate}
\end{theorem}

\begin{theorem}\label{t-star}
Let $H$ be a graph. For the class of $H$-free graphs it holds that:\\[-10pt]
\begin{enumerate}
\item [(i)] {\sc Star Colouring} is polynomial-time solvable if $H\ssi P_4$ and \NP-complete if $H\not\ssi P_4$ and $H\neq 2P_2$;\\[-10pt]
\item [(ii)] for every $k\geq 3$, {\sc Star $k$-Colouring} is polynomial-time solvable if $H$ is a linear forest and \NP-complete otherwise.
\end{enumerate}
\end{theorem}

\begin{theorem}\label{t-injective}
Let $H$ be a graph. For the class of $H$-free graphs it holds that:\\[-10pt]
\begin{enumerate}
\item [(i)] {\sc Injective Colouring} is polynomial-time solvable if $H\ssin 2P_1+P_4$ and \NP-complete if $H\not\ssi 2P_1+P_4$;\\[-10pt]
\item [(ii)] for every $k\geq 4$, {\sc Injective $k$-Colouring} is polynomial-time solvable if $H$ is a linear forest and \NP-complete otherwise.
\end{enumerate}
\end{theorem}

\noindent
In Section~\ref{s-con} we give a number of open problems that resulted from our systematic study; in particular we will discuss the distance constrained labelling framework in more detail.

\begin{center}
\begin{table}
\begin{tabular}{|l|l|l|c|}
\hline
& polynomial time & NP-complete\\
\hline
\textsc{Colouring}~\cite{KKTW01} & $H\ssi P_4$ or $P_1+P_3$ & else\\
\hline
\textsc{Acyclic Colouring} & $H\ssi P_4$ & else except for 1 open case: $H=2P_2$\\
\hline
\textsc{Star Colouring} & $H\ssi P_4$ & else except for 1 open case: $H=2P_2$\\
\hline
\textsc{Injective Colouring} & $H\ssin 2P_1+P_4$& else except for 1 open case: $H=2P_1+P_4$\\
\hline
\textsc{$k$-Colouring} (see~\cite{CHSZ18,GJPS17,KMMNPS18}) & depends on $k$ & infinitely many open cases for all $k\geq 3$\\
\hline
\textsc{Acyclic $k$-Colouring} $(k\geq 3$) & $H$ is a linear forest & else\\
\hline
\textsc{Star $k$-Colouring} $(k\geq 3)$ & $H$ is a linear forest & else\\
\hline
\textsc{Injective $k$-Colouring} $(k\geq 4)$ & $H$ is a linear forest & else\\
\hline
\end{tabular}\\[5pt]
\caption{The state-of-the-art for the three problems in this paper and the original \textsc{Colouring} problem; both when $k$ is fixed and part of the input.
The only open case for \textsc{Acyclic Colouring} and \textsc{Star Colouring} is $H=2P_2$. The only open case for \textsc{Injective Colouring} is $H=2P_1+P_4$.}\label{t-thetable}
\end{table}
\end{center}

\section{A General Polynomial Result}\label{s-gen}

A \emph{biclique} or \emph{complete bipartite graph} is a bipartite graph on vertex set $S\cup T$, such that $S$ and $T$ are independent sets and there is an edge between every vertex of $S$ and every vertex of $T$; if $|S|=s$ and $|T|=t$, we denote this graph by $K_{s,t}$ , and if $s=t$, the biclique is \emph{balanced} and of \emph{order}~$s$.
We say that a colouring $c$ of a graph $G$ satisfies the \emph{balance biclique condition} (BB-condition) if $c$ uses at least $k+1$ colours to colour $G$, where $k$ is the order of a largest biclique that is contained in $G$ as a (not necessarily induced) subgraph.

Let $\pi$ be some colouring property; e.g., $\pi$ could mean being acyclic, star or injective. Then $\pi$ \emph{can be expressed in MSO$_2$} (monadic second-order logic with edge sets) if for every $k\geq 1$, the graph property of having a $k$-colouring with property~$\pi$ can be expressed in MSO$_2$.
The general problem \textsc{Colouring$(\pi)$} is to decide, on a graph $G$ and integer $k\geq 1$, if~$G$ has a $k$-colouring with property~$\pi$.
If $k$ is fixed, we write \textsc{$k$-Colouring($\pi$)}. We now prove the following result.

\begin{theorem}\label{t-general}
Let $H$ be a linear forest, and let $\pi$ be a colouring property that can be expressed in MSO$_2$, such that every colouring with property $\pi$ satisfies the BB-condition.
Then, for every integer $k\geq 1$, $k$-{\sc Colouring($\pi$)} is linear-time solvable for $H$-free graphs.
\end{theorem}

\begin{proof}
Atminas, Lozin and Razgon~\cite{ALR12} proved that that for every pair of integers $\ell$ and~$k$, there exists a constant $b(\ell,k)$
such that every graph of treewidth at least $b(\ell,k)$ contains an induced $P_\ell$ or a (not necessarily induced) biclique $K_{k,k}$.
Let $G$ be an $H$-free graph, and let $\ell$ be the smallest integer such that $H\ssi P_\ell$; observe that $\ell$ is a constant.
Hence, we can use Bodlaender's algorithm~\cite{Bo96} to test in linear time if $G$ has treewidth at most $b(\ell,k)-1$.

First suppose that the treewidth of $G$ is at most $b(\ell,k)-1$. As $\pi$ can be expressed in MSO$_2$,
the result of Courcelle~\cite{Co90} allows us to test in linear time whether $G$ has a $k$-colouring with property~$\pi$.
Now suppose that the treewidth of $G$ is at least $b(\ell,k)$. As $G$ is $H$-free, $G$ is $P_\ell$-free.
Then, by the result of Atminas, Lozin and Razgon~\cite{ALR12}, we find that $G$ contains $K_{k,k}$ as a subgraph.
As $\pi$ satisfies the BB-condition, $G$ has no $k$-colouring with property $\pi$.
\end{proof}

\noindent
We now apply Theorem~\ref{t-general} to obtain the polynomial cases for fixed $k$ in Theorem~\ref{t-acyclic}--\ref{t-injective}.

\begin{corollary}\label{c-linearforest}
Let $H$ be a linear forest. For every $k\geq 1$, {\sc Acyclic $k$-Colouring}, {\sc Star $k$-Colouring} and {\sc Injective $k$-Colouring} are polynomial-time solvable for $H$-free graphs.
\end{corollary}

\begin{proof}
All three kinds of colourings use at least $s$ colours to colour $K_{s,s}$ (as the vertices from one bipartition class of $K_{s,s}$ must receive unique colours). Hence, every acyclic, star and injective colouring of every graph satisfies the BB-condition. Moreover, it is readily seen that the colouring properties of being acyclic, star or injective can all be expressed in MSO$_2$. Hence, we may apply Theorem~\ref{t-general}.
\end{proof}

\section{Acyclic Colouring}\label{s-acyclic}

In this section, we prove Theorem~\ref{t-acyclic}. For a graph $G$ and a colouring $c$, the pair $(G,c)$ has a \emph{bichromatic} cycle~$C$ if $C$ is a cycle of $G$ with
$|c(V(C)|=2$, that is, the vertices of $C$ are coloured by two alternating colours (so $C$ is even).
The notion of a \emph{bichromatic path} is defined in a similar matter.

\begin{lemma}\label{l-girth}
For every $k\geq 3$ and every $g\geq 3$, {\sc Acyclic $k$-Colouring} is \NP-complete for $2$-degenerate bipartite graphs of girth at least $g$.
\end{lemma}

\begin{proof}
We reduce from \textsc{Acyclic $k$-Colouring}, which is known to be \NP-complete for bipartite graphs for every $k\geq 3$~\cite{CC86}.
Recall that the \emph{arboricity} of a graph is the minimum number of forests needed to partition its edge set.
By counting the edges, a graph with arboricity at most $t$ is $(2t-1)$-degenerate and thus $2t$-colourable.
We start by taking a graph $F$ that has no $2k(k-1)$-colouring and that is of girth at least~$g$.
By a seminal result of Erd\H{o}s~\cite{Er59}, such a graph $F$ exists (and its size is constant, as it only depends on $g$ and $k$ which are fixed integers).
Notice that $F$ does not admit a vertex-partition into $k$ subgraphs with arboricity at most $k-1$, since otherwise $F$ would be $2k(k-1)$-colourable.

Now we consider the graph~$S$ obtained by subdividing every edge of $F$ exactly once.
The graph $S$ is $2$-degenerate and bipartite with the \emph{old} vertices from $F$ in one part and the \emph{new} vertices of degree~$2$ in the other part.
Moreover, $S$ has girth at least~$g$, as $F$ has girth at least~$g$.

We claim that $S$ has no acyclic $k$-colouring.
For contradiction, assume that $S$ has an acyclic $k$-colouring. Assign the colour of every old vertex to the corresponding vertex of $F$
and assign the colour of every new vertex to the corresponding edge of $F$.
For every colour $i$, we consider the subgraph $F_i$ of $F$ induced by the vertices coloured $i$.
For every $j\neq i$, the subgraph of $S$ induced by the colours $i$ and $j$ is a forest.
This implies that the subgraph of $F_i$ induces by the edges coloured $j$ is a forest.
So the arboricity of $F_i$ is at most $k-1$, that is, the number of colours distinct from $i$.
By previous discussion, the chromatic number of $F_i$ is at most $2(k-1)$, so that $F$ is $2k(k-1)$-colourable.
This contradiction shows that $S$ has no acyclic $k$-colouring.

We repeatedly remove new vertices from $S$ until we obtain a graph $S'$ that is acyclically $k$-colourable.
Note that $S'$ has girth at least~$g$ 
and is $2$-degenerate, as $S$ has girth at least~$g$ and is $2$-degenerate.
Let $x_2$ be the last vertex that we removed and let $x_1$ and $x_3$ be the neighbours of $x_2$ in $S$.
By construction, $S'$ is acyclically $k$-colourable and every acyclic $k$-colouring $c$ of $S'$ is such that:

\begin{itemize}
 \item $c(x_1)=c(x_3)$, since otherwise setting $c(x_2)\not\in\{c(x_1),c(x_3)\}$ would extend $c$ to an acyclic $k$-colouring of the larger graph, which is not possible by construction. Without loss of generality, $c(x_1)=c(x_3)=1$.
 \item For every colour $i\neq 1$, $S'$ contains a bichromatic path coloured $1$ and $i$ between $x_1$ and $x_3$, since otherwise setting $c(x_2)=i$ would extend $c$ to an acyclic $k$-colouring of the larger graph again.
\end{itemize}

\noindent
We are ready to describe the reduction. Let $G$ be a bipartite instance of \textsc{Acyclic $k$-Colouring}.
We construct an equivalent instance $G'$ with girth at least~$g$ as follows.
For every vertex $z$ of $G$, we fix an arbitrary order on the neighbours of $z$.
We replace $z$ of $G$ by $d$ vertices $z_1,z_2,\ldots,z_d$, where $d$ is the degree of $z$.
Then for $1\leq i\leq d-1$, we take a copy of $S'$ and we identify the vertex $x_1$ of $S'$ with $z_i$ and the vertex $x_3$ of $S'$ with $z_{i+1}$.
Now for every edge $uv$ of $G$, say $v$ is the $i^{th}$ neighbour of $u$ and $u$ is the $j^{th}$ neighbour of $v$, we add the edge 
$u_iv_j$ in $G'$. See also Figure~\ref{f-latest}.

Given an acyclic $k$-colouring of $G$, we assign the colour of $z$ to $z_1,\ldots,z_d$ and extend the colouring to the copies of $F'$, which gives an acyclic colouring of $G'$.
Given an acyclic $k$-colouring of $G'$, the copies of $F'$ force the same colour on $z_1,\ldots,z_d$ and we assign this common colour to $z$, which gives an acyclic $k$-colouring of $G$.

Finally, notice that since $G$ and $S'$ are bipartite, $G'$ is bipartite.
As $S'$ is $2$-degenerate and has girth at least~$g$, we find that $G'$ is $2$-degenerate and has girth at least $g$.
\end{proof}

\begin{figure}[h]
\centering
\includegraphics[width=1\textwidth]{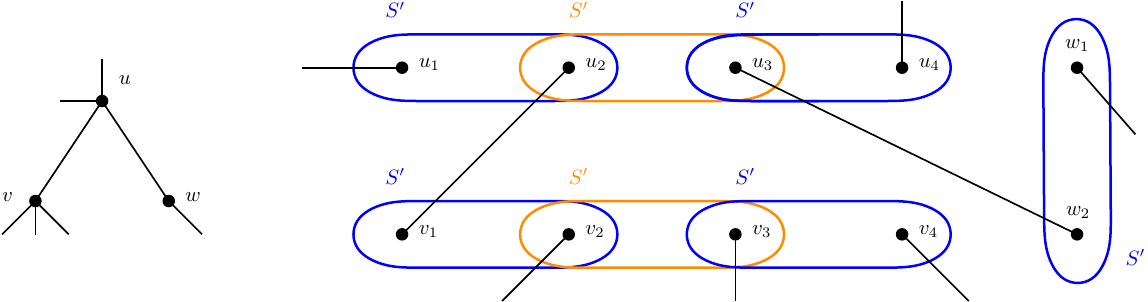}
\caption{An example of part of a graph $G$ (left) and the corresponding part in $G'$ (right). In the part of $G'$ corresponding to vertex $u$, vertex~$u_1$ is identified with $x_1$ of the left copy of $S'$; vertex $u_2$ with $x_3$ of the left copy of $S'$ and $x_1$  of the middle copy of $S'$; vertex $u_3$ with $x_3$ of the middle copy of $S'$ and $x_1$ of the right copy of $S'$; and $u_4$  with $x_3$ of the right copy of $S'$.}
\label{f-latest}
\end{figure}

\noindent
The \emph{line graph} of a graph $G$ has vertex set $E(G)$ and an edge between two vertices $e$ and $f$ if and only if $e$ and $f$ share an end-vertex of $G$.
We now modify the construction of~\cite{AZ02} for line graphs from $k=3$ to $k\geq 3$.

\begin{lemma}\label{l-az}
For every $k\geq 3$, {\sc Acyclic $k$-Colouring} is \NP-complete for line graphs of multigraphs.
\end{lemma}

\begin{proof}
For an integer $k\geq 1$, a \emph{$k$-edge colouring} of a graph $G=(V,E)$ is a mapping $c:E\to \{1,\ldots,k\}$ such that $c(e)\neq c(f)$ whenever the edges $e$ and $f$ share an end-vertex. A \emph{colour class} consists of all edges of $G$ that are mapped by $c$ to a specific colour~$i$.
For a fixed integer~$k\geq 1$, the \textsc{Acyclic $k$-Edge Colouring} problem is to decide if a given graph has an acyclic $k$-edge colouring. Alon and Zaks proved that \textsc{Acyclic $3$-Edge Colouring} is \NP-complete. We note that a graph has an acyclic $k$-edge colouring if and only if its line graph has an acyclic $k$-colouring. Hence, it remains to generalize the construction of Alon and Zaks~\cite{AZ02} from $k=3$ to $k\geq 3$. Our main tool is the gadget graph $F_k$, depicted in Figure~\ref{fig:clawgadget}, about which we prove the following two claims.

\medskip
\noindent
\emph{(i) The edges of $F_k$ can be coloured acyclically using $k$ colours, with no bichromatic path between $v_1$ and $v_{14}$.}

\smallskip
\noindent
\emph{(ii) Every acyclic $k$-edge colouring of $F_k$ using $k$ colours assigns $e_1$ and $e_2$ the same colour.}

\begin{figure}[h]
\centering
\includegraphics[width=0.7\textwidth]{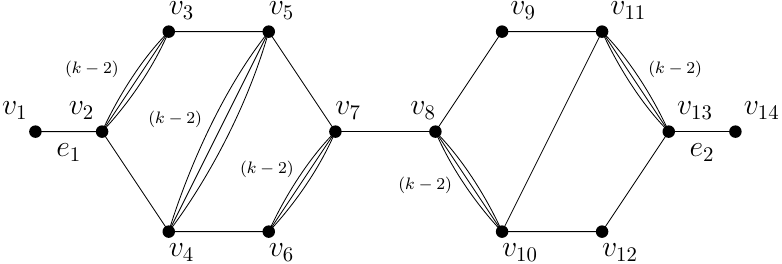}
\caption{The gadget multigraph $F_k$. The labels on edges are multiplicities.}
\label{fig:clawgadget}
\end{figure}

\smallskip
\noindent
We first prove (ii). We assume, without loss of generality, that $v_1v_2$ is coloured by $1$, $v_2v_4$ by $2$ and the edges between $v_2$ and $v_3$ by colours $3,\ldots, k$.
The edge $v_3v_5$ has to be coloured by $1$, otherwise we have a bichromatic cycle on $v_2v_3v_5v_4$. This necessarily implies that

\begin{itemize}
\item the edges between $v_4$ and $v_5$ are coloured by $3,\ldots, k$,
\item the edge $v_5v_7$ is coloured by $2$,
\item the edge $v_4v_6$ is coloured by $1$,
\item the edges between $v_6$ and $v_7$ are coloured by $3,\ldots, k$, and
\item the edge $v_7v_8$ is coloured by $1$.
\end{itemize}

\noindent
Now assume that the edge $v_8v_9$ is coloured by $a \in \{2,\ldots,k\}$ and the edges between $v_8$ and $v_{10}$ by colours from the set $A$, where $A =\{2,\ldots,k\} \setminus a$. The edge $v_{10}v_{11}$ is either coloured $a$ or $1$. However, if it is coloured $1$, $v_9v_{11}$ is assigned a colour $b \in A$ and necessarily we have either a bichromatic cycle on $v_8v_9v_{11}v_{13}v_{12}v_{10}$, coloured by $b$ and $a$, or a bichromatic cycle on $v_{10}v_{11}v_{13}v_{12}$, coloured by $a$ and $1$. Thus $v_{10}v_{11}$ is coloured by $a$. To prevent a bichromatic cycle on $v_8v_9v_{11}v_{10}$, the edge $v_9v_{11}$ is assigned colour $1$. The rest of the colouring is now determined as $v_{10}v_{12}$ has to be coloured by $1$, the edges between $v_{11}$ and $v_{13}$ by $A$, $v_{12}v_{13}$ by $a$, and $v_{13}v_{14}$ by $1$. We then have a $k$-colouring with no bichromatic cycles of size at least $3$ in $F_k$ for every possible choice of $a$. This proves that $v_1v_2$ and $v_{13}v_{14}$ are coloured alike under every acyclic $k$-edge colouring. 

We prove (i) by choosing $a$ different from $2$. Then there is no bichromatic path between $v_1$ and $v_{14}$.

We now reduce from \textsc{$k$-Edge-Colouring} to \textsc{Acyclic $k$-Edge Colouring} as follows. Given an instance $G$ of \textsc{$k$-edge Colouring} we construct an instance $G^{\prime}$ of \textsc{Acyclic $k$-Edge Colouring} by replacing each edge $uv$ in $G$ by a copy of $F_k$ where $u$ is identified with $v_1$ and $v$ is identified with $v_{14}$.

If $G^{\prime}$ has an acyclic $k$-edge colouring $c^{\prime}$ then we obtain a $k$-edge colouring $c$ of $G$ by setting $c(uv)=c^{\prime}(e_1)$ where $e_1$ belongs to the gadget $F_k$ corresponding to the edge $uv$. If $G$ has a $k$-edge colouring $c$ then we obtain an acyclic $k$-edge colouring $c^{\prime}$ of $G^{\prime}$ by setting $c^{\prime}(e_1)=c(uv)$ where $e_1$ belongs to the gadget corresponding to the edge $uv$. The remainder of each gadget $F_k$ can then be coloured as described above.
\end{proof}

\noindent
In our next result, $k$ is part of the input.
Recall that a graph is \emph{co-bipartite} if it is the complement of a bipartite graph.
As bipartite graphs are $C_3$-free, co-bipartite graphs are $3P_1$-free.

\begin{lemma}\label{l-3p1acyclic}
{\sc Acyclic Colouring} is \NP-complete for co-bipartite graphs.
\end{lemma}
\begin{proof}
Alon et al.~\cite[Theorem 3.5]{ACGMRSWY20} proved that deciding if a balanced bipartite graph on $2n$ vertices has a connected matching of size $n$ is \NP-complete. A matching is called \emph{connected} if no two edges of the matching induce $2K_2$ in the given graph. We shall reduce from this problem to prove our theorem.

To this end, we claim that a balanced bipartite graph $G$ with parts $A$ and $B$ such that $|A|=|B|=n$ has a connected matching of size $n$ if and only if its complement has an acyclic colouring with $n$ colours.

Suppose that there is an acyclic colouring $c$ of $\overline{G}$ with $n$ colours. Clearly, such colouring uses $n$ colours on $A$ and $n$ colours on $B$. Vertices coloured with the same colour do not have an edge between them in $\overline{G}$ and thus are connected by an edge in $G$. Let us take the set of edges formed by each of the $n$ colour classes. By the property of colouring, this is a matching in $G$ and it is of size $n$. To see that it is also connected, suppose for a contradiction that there are two edges of the matching, say $a_1b_1$ and $a_2b_2$, forming an induced $2K_2$ in $G$. Without loss of generality, $c(a_1)=c(b_1)=1$ and $c(a_2)=c(b_2)=2$.
Now the induced $2K_2$ in $G$ corresponds to a 4-cycle in $\overline{G}$ coloured with two colours, a contradiction with $c$ being an acyclic colouring.
 
In the opposite direction, let us have a connected matching of size $n$ in $G$. Colour the $n$ vertices in $A$ by $1,\ldots,n$. Let us colour the vertices of $B$ with respect to the connected matching so that each vertex of $B$ gets the colour of the vertex in $A$ it is matched to. Indeed, this is a colouring of $\overline{G}$ by $n$ colours. It remains to prove that it is acyclic. Any cycle in $G$ having more than five vertices has by the definition of our colouring at least three colours. Therefore, a possible bichromatic cycle in $\overline{G}$ must be of size $4$. The only possibility for such 4-cycle is that it is formed by two pairs of vertices, each one forming an edge of the connected matching in $G$. However, this implies that these two matching edges induce $2K_2$ in $G$, a contradiction with the connectedness of the original matching. This finishes the proof our claim.
\end{proof}

\noindent
We combine the above results with a result of Lyons~\cite{Ly11} to prove Theorem~\ref{t-acyclic}.

\medskip
\noindent
{\bf Theorem~\ref{t-acyclic} (restated).}
\emph{Let $H$ be a graph. For the class of $H$-free graphs it holds that:\\[-10pt]
\begin{enumerate}
\item [(i)] {\sc Acyclic Colouring} is polynomial-time solvable if $H\ssi P_4$ and \NP-complete if $H\not\ssi P_4$ and $H\neq 2P_2$;\\[-10pt]
\item [(ii)] for every $k\geq 3$, {\sc Acyclic $k$-Colouring} is polynomial-time solvable if $H$ is a linear forest and \NP-complete otherwise.
\end{enumerate}}

\begin{proof}
We first prove (ii). First suppose that $H$ contains an induced cycle $C_p$. Then we use Lemma~\ref{l-girth}.
Now assume $H$ has no cycle so $H$ is a forest. If $H$ has a vertex of degree at least~$3$, then $H$ has an induced $K_{1,3}$.
As every line graph of a multigraph is $K_{1,3}$-free, we can use Lemma~\ref{l-az}. Otherwise $H$ is a linear forest and we use Corollary~\ref{c-linearforest}. 

We now prove (i). Due to (ii), we may assume that $H$ is a linear forest.
If $H\ssi P_4$, then we use the result of Lyons~\cite{Ly11} that states that \textsc{Acyclic Colouring} is polynomial-time solvable for $P_4$-free graphs.
Now suppose $3P_1\ssi H$. By Lemma~\ref{l-3p1acyclic}, \textsc{Acyclic Colouring} is \NP-complete for co-bipartite graphs and thus for $3P_1$-free graphs. 
It remains to consider the case where $H=2P_2$, but this case was excluded from the statement of the theorem. 
\end{proof}

\section{Star Colouring}\label{s-star}

In this section we prove Theorem~\ref{t-star}. We first prove the following lemma.

\begin{figure}
\centering
\includegraphics[width=100mm]{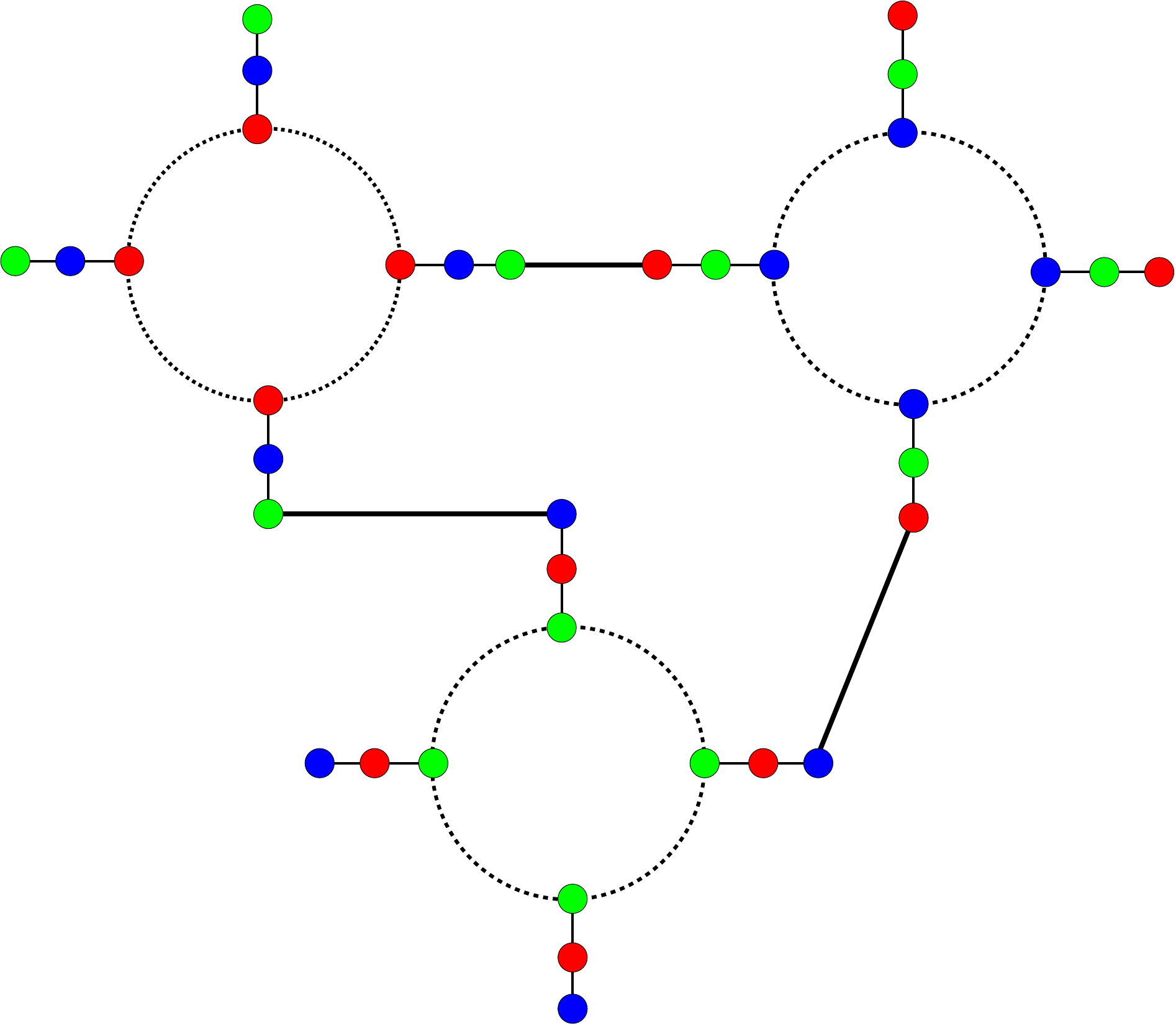}
\caption{A star $3$-colouring of the graph $G'$ obtained from a $3$-colouring of $G=K_3$. Only part of $G'$ is displayed.}
\label{fig:star3}
\end{figure}

\begin{lemma} \label{lem:star-col-high-girth}
For every $g \geq 3$, {\sc Star $3$-Colouring} is \NP-complete for planar graphs of girth at least $g$ and maximum degree $3$.
\end{lemma}

\begin{proof}
We reduce from {\sc $3$-Colouring}, which is \NP-complete even for planar graphs with maximum degree 4~\cite{GJS76}.
Let $G$ be an instance of this restricted version of \textsc{$3$-Colouring}.
The vertex gadget $V$ contains

\begin{itemize}
 \item a cycle of length $12g$ with vertices $d_1,\ldots,d_{12g}$,
 \item $12g$ independent vertices $e_1,\ldots,e_{12g}$ such that $e_i$ is adjacent to $d_i$ for every $1\leq i\leq 12g$, and
 \item four independent vertices $f_1,f_2,f_3,f_4$ such that $f_i$ is adjacent to $e_{3ig}$ for every $1\leq i\leq 4$.
\end{itemize}

\noindent
We construct an instance $G'$ of \textsc{Star $3$-Colouring} from $G$ as follows.
We consider a planar embedding of $G$ and for every vertex $x$, we order the neighbours of $x$ in a clockwise way. 
Then we replace $x$ by a copy $V_x$ of $V$.
Now for every edge $mn$ of $G$, say $n$ is the $i^{th}$ neighbour of $m$ and $m$ is the $j^{th}$ neighbour of $n$,
we add the edge between the vertex $f_i$ of $V_m$ and the vertex $f_j$ of $V_n$, see Figure~\ref{fig:star3}.

It is not hard to check that in every star $3$-colouring of $V$, the four vertices $f_i$ get the same colour.
Moreover, there is no bichromatic path between any two vertices $f_i$.

Suppose that $G$ admits a $3$-colouring $c$ of with colours in $\{0,1,2\}$.
For every vertex $x$ in~$G$, we assign $c(x)$ to the vertices $f_i$ in $V_x$ and we assign $(c(x)+1)\pmod 3$ to the vertices~$e_{3ig}$.
Then we extend this pre-colouring into a star $3$-colouring of $V_x$.
This gives a star $3$-colouring of $G'$.
Given a star $3$-colouring of $G'$, we assign to every vertex $x$ in $G$ the colour of the vertices $f_i$ in $V_x$, which gives a $3$-colouring of $G$.

Finally, as $G$ is planar with maximum degree $4$, it holds that $G'$ is planar with maximum degree $3$.
Moreover, by construction, $G'$ has girth at least $g$.
\end{proof}

\noindent
Now we begin our development for Theorem~\ref{t-star}.

\begin{lemma}\label{l-evencycle}
Let $p\geq 4$ be a fixed integer. Then, for every $k\geq 3$, {\sc Star $k$-Colouring} is \NP-complete for $C_p$-free graphs.
\end{lemma}

	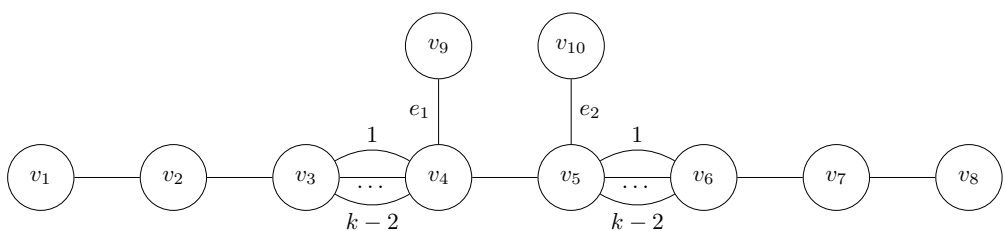
\begin{figure}[b]
		\resizebox{!}{3cm}{\begin{tikzpicture}[main_node/.style={circle,draw,minimum size=1cm,inner sep=3pt]}]

		\node[main_node](v1) at (0,0){$v_1$};
		\node[main_node](v2) at (2,0){$v_2$};
		\node[main_node](v3) at (4,0){$v_3$};
		\node[main_node](v4) at (6,0){$v_4$};
		\node[main_node](v5) at (8,0){$v_5$};
		\node[main_node](v6) at (10,0){$v_6$};
		\node[main_node](v7) at (12,0){$v_7$};
		\node[main_node](v8) at (14,0){$v_8$};
		\node[main_node](v9) at (6,2){$v_9$};
		\node[main_node](v10) at (8,2){$v_{10}$};

		\draw(v1)--(v2)--(v3);
		\draw (v3) to node[below] {$\dots$} (v4);
		\draw (v4)--(v5);
		\draw(v6)--(v7)--(v8);
		\draw (v3) to[out=30, in=150] node[above] {$1$} (v4); 
		\draw (v3) to[out=330, in=210] node[below] {$k-2$} (v4);
		\draw (v5) to[out=30, in=150] node[above] {$1$} (v6); 
		\draw (v5) to[out=330, in=210] node[below] {$k-2$} (v6);
		\draw(v5) to node[below] {$\dots$} (v6);
		\draw(v4) to node[left] {$e_1$}(v9);
		\draw(v5) to node[right] {$e_2$}(v10);
		\end{tikzpicture}}
		\caption{The gadget $F_k$ in the proof of Lemma~\ref{l-linestar}.}\label{staredge}
	\end{figure}

\begin{proof}
The case $k=3$ follows from Lemma~\ref{lem:star-col-high-girth}.
We obtain \NP-completeness for $k\geq 4$ by a reduction from \textsc{Star $3$-Colouring} for $C_p$-free graphs by adding a dominating clique of size $k-3$.
\end{proof}

\noindent
In Lemma~\ref{l-linestar} we extend the recent result of Lei et al.~\cite{LSS18} from $k=3$ to $k\geq 3$.

\begin{lemma}\label{l-linestar}
For every $k\geq 3$, {\sc Star $k$-Colouring} is \NP-complete for line graphs of multigraphs.
\end{lemma}

\begin{proof}
Recall that for an integer $k\geq 1$, a \emph{$k$-edge colouring} of a graph $G=(V,E)$ is a mapping $c:E\to \{1,\ldots,k\}$ such that $c(e)\neq c(f)$ whenever
the edges $e$ and $f$ share an end-vertex. Recall also that the notions of a colour class and bichromatic subgraph for colourings has its natural analogue for edge colourings.
A proper edge $k$-colouring $c$ is a \emph{star} $k$-edge colouring if the union of any two colour classes does not contain a path or cycle of on four edges.
For a fixed integer~$k\geq 1$, the \textsc{Star $k$-Edge Colouring} problem is to decide if a given graph has a star $k$-edge colouring.
Lei et al.~\cite{LSS18} proved that \textsc{Star $3$-Edge Colouring} is \NP-complete.
Dvo\v{r}{\'{a}}k et al.~\cite{DMS13} observed that a graph has a star $k$-edge colouring if and only if its line graph has a star $k$-colouring.
Hence, it suffices to follow the proof of Lei et al.\cite{LSS18} in order to generalize the case $k=3$ to $k\geq 3$.
As such, we give a reduction from \textsc{$k$-Edge Colouring} to \textsc{Star $k$-Edge Colouring} which makes use of the gadget $F_k$ in Figure~\ref{staredge}.
First we consider separately the case where the edges $e_1=v_4v_9$ and $e_2=v_5v_{10}$ are coloured alike and the case where they are coloured differently to show that in any star $k$-edge colouring of the gadget $F_k$ shown in Figure~\ref{staredge}, $v_1v_2$ and $v_7v_8$ are assigned the same colour.

Assume $c(e_1)=c(e_2)=1$. We may then assume that the edge $v_4v_5$ is assigned colour $2$ and the remaining $k-2$ colours are used for the multiple edges $v_3v_4$ and $v_5v_6$. The edge $v_2v_3$, and similarly $v_6v_7$, must then be assigned colour $1$ to avoid a bichromatic $P_5$ on the vertices $\{ v_2, v_3,v_4, v_5, v_6\}$ using any two of the multiple edges in a single colour. The edge $v_1v_2$, and similarly $v_7v_8$ must then be assigned colour $2$ to avoid a bichromatic $P_5$ on the vertices $\{v_1, v_2, v_3, v_4, v_9\}$.

Next assume $e_1$ and $e_2$ are coloured differently. Without loss of generality, let $c(e_1)=1$, $c(e_2)=2$ and $c(v_4v_5)=3$. The multiple edges $v_3v_4$ must then be assigned colours $2$ and $4 \dots k$ and $v_5v_6$ colour $1$ and colours $4 \dots k$. To avoid a bichromatic $P_5$ on the vertices $\{v_2,v_3,v_4,v_5,v_6\}$, $v_2v_3$ must be coloured $1$. Similarly, $v_6v_7$ must be assigned colour $2$. Finally, to avoid a bichromatic $P_5$ on the vertices $\{v_1,v_2,v_3,v_4,v_9\}$, $v_1v_2$ must be coloured $3$. By a similar argument, $v_7v_8$ must also be coloured $3$, hence $v_1v_2$ and $v_7v_8$ must be coloured alike.

We can then replace every edge $e$ in some instance $G$ of \textsc{$k$-Edge-Colouring} by a copy of $F_k$, identifying its endpoints with $v_1$ and $v_8$, to obtain an instance $G^{\prime}$ of \textsc{Star $k$-Edge-Colouring}. If $G$ is $k$-edge-colourable we can star $k$-edge-colour $G^{\prime}$ by setting $c^{\prime}(v_1v_2)=c^{\prime}(v_7v_8)=c(e)$. If $G^{\prime}$ is star $k$-edge-colourable, we obtain a $k$-edge-colouring of $G$ by setting $c(e)=c^{\prime}(v_1v_2)$.
\end{proof}

\noindent
We now combine the above results 
with results of Albertson et al.~\cite{ACKKR04}, Lyons~\cite{Ly11} and Shalu and Anthony~\cite{SA20} to prove Theorem~\ref{t-star}.

\medskip
\noindent
{\bf Theorem~\ref{t-star} (restated).}
\emph{Let $H$ be a graph. For the class of $H$-free graphs it holds that:\\[-10pt]
\begin{enumerate}
\item [(i)] {\sc Star Colouring} is polynomial-time solvable if $H\ssi P_4$ and \NP-complete if $H\not\ssi P_4$ and $H\neq 2P_2$;\\[-10pt]
\item [(ii)] for every $k\geq 3$, {\sc Star $k$-Colouring} is polynomial-time solvable if $H$ is a linear forest and \NP-complete otherwise.
\end{enumerate}}

\begin{proof}
We first prove (ii). First suppose that $H$ contains an induced odd cycle.
Then the class of bipartite graphs is contained in the class of $H$-free graphs.
Lemma 7.1 in Albertson et al.~\cite{ACKKR04} implies, together with the fact that for every $k\geq 3$, $k$-\textsc{Colouring} is \NP-complete,
that for every $k\geq 3$, \textsc{Star $k$-Colouring} is \NP-complete for bipartite graphs. If $H$ contains an induced even cycle, then we use Lemma~\ref{l-evencycle}.
Now assume $H$ has no cycle, so $H$ is a forest. If $H$ contains a vertex of degree at least~$3$, then $H$ contains an induced $K_{1,3}$.
As every line graph of a multigraph is $K_{1,3}$-free, we can use Lemma~\ref{l-linestar}. Otherwise $H$ is a linear forest, in which case we use Corollary~\ref{c-linearforest}. 

We now prove (i). Due to (ii), we may assume that $H$ is a linear forest.
If $H\ssi P_4$, then we use the result of Lyons~\cite{Ly11} that states that \textsc{Star Colouring} is polynomial-time solvable for $P_4$-free graphs.
Now suppose $3P_1\ssi H$. Shalu and Antony~\cite{SA20} who proved that \textsc{Star Colouring} is \NP-complete for co-bipartite graphs and thus for $3P_1$-free graphs. 
It remains to consider the case where $H=2P_2$, but this case was excluded from the statement of the theorem. 
\end{proof}

\section{Injective Colouring}\label{s-injective}

In this section we prove Theorem~\ref{t-injective}. 
We first show a hardness result for fixed $k$.\footnote{We note that Janczewski et al.~\cite{JKM09} proved that \textsc{$L(p,q)$-Labelling} is \NP-complete for planar bipartite graphs, but in their paper they assumed that $p>q$.}

\begin{lemma}\label{l-triangle}
For every $k\geq 4$, {\sc Injective $k$-Colouring} is \NP-complete for bipartite graphs.
\end{lemma}

\begin{proof}
We reduce from \textsc{Injective $k$-Colouring}; recall that this problem is \NP-complete for every $k\geq 4$.
Let $G=(V,E)$ be a graph. We construct a graph $G^{\prime}$ as follows. For each edge $uv$ of $G$, we remove the edge $uv$ and add two vertices $u^{\prime}_v$, which we make adjacent to $u$, and $v^{\prime}_u$, which we make adjacent to $v$. 
Next, we place an independent set $I_{uv}$ of $k-2$ vertices adjacent to both $u^{\prime}_v$ and $v^{\prime}_u$. 
Note that $G^\prime$ is bipartite: we can let one partition class consist of all vertices of $V(G)$ and the vertices of the $I_{uv}$-sets and the other one consist of all the remaining vertices (that is, all the ``prime'' vertices we added). 
It remains to show that $G^{\prime}$ has an injective $k$-colouring if and only if $G$ has an injective $k$-colouring.

\begin{figure}[h]
\centering
\includegraphics[width=0.5\textwidth]{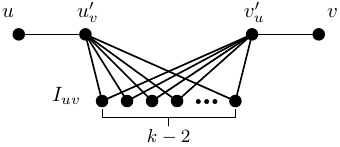}
\caption{The edge gadget used in the proof of Lemma~\ref{l-triangle}.}
\label{fig:l-triangle}
\end{figure}

First assume that $G$ has an injective $k$-colouring $c$. Colour the vertices of $G^{\prime}$ corresponding to vertices of $G$ as they are coloured by $c$. We can extend this to an injective $k$-colouring $c^{\prime}$ of $G^{\prime}$ by considering the gadget corresponding to each edge $uv$ of $G$. Set $c^{\prime}(u^{\prime}_v)=c^{\prime}(v)$ and $c^{\prime}(v^{\prime}_u)=c^{\prime}(u)$. We can now assign the remaining $k-2$ colours to the vertices of the independent sets. Clearly $c^{\prime}$ creates no bichromatic $P_3$ involving vertices in at most one edge gadget. Assume there exists a bichromatic $P_3$ involving vertices in more than one edge gadget, then this path must consist of a vertex $u$ of $G$ together with two gadget vertices $u^{\prime}_v$ and $u^{\prime}_w$ which are coloured alike. This is a contradiction since it implies the existence of a bichromatic path $v,u,w$ in $G$.

Now assume that $G^{\prime}$ has an injective $k$-colouring $c^{\prime}$. Let $c$ be the restriction of $c^{\prime}$ to those vertices of $G^{\prime}$ which correspond to vertices of $G$. To see that $c$ is an injective colouring of $G$, note that we must have $c^{\prime}(u^{\prime}_v)=c^{\prime}(v)$ and $c^{\prime}(v^{\prime}_u)=c^{\prime}(u)$ for any edge $uv$. Therefore, if $c$ induces a bichromatic $P_3$ on $u,v,w$, then $c^{\prime}$ induces a bichromatic $P_3$ on $v^{\prime}_u, v, v^{\prime}_w$. We conclude that $c$ is injective.
\end{proof}

\noindent
We now turn to the case where $k$ is part of the input and first prove a number of positive results. The \emph{complement} of a graph $G$ is the graph $\overline{G}$ with vertex set $V(G)$ and an edge between two vertices $u$ and $v$ if and only if $uv\notin E(G)$. An injective colouring $c$ of a graph $G$ is \emph{optimal} if $G$ has no injective colouring using fewer colours than~$c$.
An injective colouring $c$ is \emph{$\ell$-injective} if every colour class of~$c$ has size at most~$\ell$. An $\ell$-injective colouring $c$ of a graph $G$ is \emph{$\ell$-optimal} if $G$ has no $\ell$-injective colouring that uses fewer colours than $c$. We start with a useful lemma for the case where $\ell=2$ that we will also use in our proofs.

\begin{lemma}\label{l-easy}
An optimal $2$-injective colouring of a graph $G$ can be found in polynomial time.
\end{lemma}

\begin{proof}
Let $c$ be a $2$-injective colouring of $G$. Then each colour class of size~$2$ in $G$ corresponds to a \emph{dominating} edge of $\overline{G}$ (an edge $uv$ of a graph is dominating if every other vertex in the graph is adjacent to at least one of $u,v$).
Hence, the end-vertices of every non-dominating edge in $\overline{G}$ have different colours in $G$. Algorithmically, this means we may delete every non-dominating edge of $\overline{G}$ from~$\overline{G}$; note that we do not delete the end-vertices of such an edge.

Let $\mu^*$ be the size of a maximum matching in the graph obtained from $\overline{G}$ after deleting all non-dominating edges of $\overline{G}$. The edges in such a matching will form exactly the colour classes of size $2$ of an optimal $2$-injective colouring of $G$. Hence,
the injective chromatic number of $G$ is equal to $\mu^*+(|V(G)|-2\mu^*)$. It remains to observe that we can find a maximum matching in a graph in polynomial time by using a standard algorithm.
\end{proof}

\noindent
We can now prove our first positive result.

\begin{lemma}\label{l-p1p4}
{\sc Injective Colouring} is polynomial-time solvable for $(P_1+P_4$)-free graphs.
\end{lemma}

\begin{proof}
Let $G$ be a $(P_1+P_4)$-free graph.
Since connected $P_4$-free graphs have diameter at most $2$, no two vertices can be coloured alike in an injective colouring.
Hence, the injective chromatic number of a $P_4$-free graph is equal to the number of its vertices.
Consequently, \textsc{Injective Colouring} is polynomial-time solvable for $P_4$-free graphs.
From now on, we assume that $G$ is not $P_4$-free.

We first show that any colour class in any injective colouring of $G$ has size at most $2$.
For contradiction, assume that $c$ is an injective colouring of $G$ such that there exists some colour, say colour~$1$, that has a colour class of size at least~$3$.
Let $P=x_1x_2x_3x_4$ be some induced $P_4$ of $G$.

We first consider the case where colour~$1$ appears at least twice on $P$. As no vertex has two neighbours coloured with the same colour, the only way in which this can happen is when $c(x_1)=c(x_4)=1$. By our assumption, $G-P$ contains a vertex~$u$ with $c(u)=1$. As $G$ is $(P_1+P_4)$-free, $u$ has a neighbour on $P$. As every colour class is an independent set, this means that $u$ must be adjacent to at least one of $x_2$ and $x_3$. Consequently, either $x_2$ or $x_3$ has two neighbours with colour~$1$, a contradiction.

Now we consider the case where colour~$1$ appears exactly once on $P$, say $c(x_h)=1$ for some $h\in \{1,2,3,4\}$. Then, by our assumption, $G-P$ contains two vertices $u_1$ and $u_2$ with colour~$1$. As $G$ is $(P_1+P_4)$-free, both $u_1$ and $u_2$ must be adjacent to at least one vertex of $P$, say $u_1$ is adjacent to $x_i$ and $u_2$ is adjacent to $x_j$. Then $x_i\neq x_j$, as otherwise $G$ has a vertex with two neighbours coloured~$1$.
As every colour class is an independent set, we have that $x_h\notin \{x_i,x_j\}$, and hence, $x_h$, $x_i$, $x_j$ are distinct vertices. Moreover, $x_h$ is not a neighbour of $x_i$ or $x_j$, as otherwise $x_i$ or $x_j$ has two neighbours coloured~$1$. Hence, we may assume without loss of generality that $h=1$, $i=3$ and $j=4$. As every colour class is an independent set, $u_1$ and $u_2$ are non-adjacent. However, now $\{x_1,u_1,x_3,x_4,u_2\}$ induces a $P_1+P_4$, a contradiction.

Finally, we consider the case where colour~$1$ does not appear on $P$. Let $u_1$, $u_2$, $u_3$ be three vertices of $G-P$ coloured~$1$. As before, $\{u_1,u_2,u_3\}$ is an independent set and each $u_i$ has a different neighbour on $P$. We first consider the case where $x_1$ or $x_4$, say $x_4$ is not adjacent to any $u_i$. Then we may assume without loss of generality that $u_1x_1$ and $u_2x_2$ are edges. However, now $\{x_4,u_1,x_1,x_2,u_2\}$ induces a $P_1+P_4$, which is not possible. Hence, we may assume without loss of generality that $u_1x_1$, $u_2x_2$ and $u_4x_4$ are edges of $G$. Again we find that $\{x_4,u_1,x_1,x_2,u_2\}$ induces a $P_1+P_4$, a contradiction.
 
From the above, we find that each colour class in an injective colouring of $G$ has size at most $2$. This means we can use Lemma~\ref{l-easy}. 
\end{proof}

\noindent
We use the next lemma in the proofs of the results for $H=2P_1+P_3$ and $H=3P_1+P_2$.

\begin{lemma}\label{l-4p1}
{\sc Injective Colouring} is polynomial-time solvable for $4P_1$-free graphs.
\end{lemma}

\begin{proof}
Let $G=(V,E)$ be a $4P_1$-free graph on $n$ vertices. We first analyze the structure of injective colourings of $G$.
Let $c$ be an 
optimal injective colouring of~$G$. As $G$ is $4P_1$-free, every colour class of $c$ has size at most~$3$. From all optimal injective colourings, we choose $c$ such that the number of size-$3$ colour classes is as small as possible.
We say that $c$ is \emph{class-$3$-optimal}.

Suppose $c$ contains a colour class of size~$3$, say colour~$1$ appears on three distinct vertices $u_1$, $u_2$ and $u_3$ of $G$.
As $G$ is $4P_1$-free, $\{u_1,u_2,u_3\}$ dominates $G$. As $c$ is injective, this means that every vertex in $G-\{u_1,u_2,u_3\}$ is adjacent to exactly one vertex of $\{u_1,u_2,u_3\}$. Hence, we can partition
$V\setminus \{u_1,u_2,u_3\}$ into three sets $T_1$, $T_2$ and $T_3$, such that for $i\in \{1,2,3\}$, every vertex of $T_i$ is adjacent to $u_i$ and not to any other vertex of $\{u_1,u_2,u_3\}$. 
If two vertices $t,t'$ in the same $T_i$, say $T_1$, are non-adjacent, then $\{t,t',u_2,u_3\}$ induces a $4P_1$, a contradiction. Hence, we partitioned $V$ into three cliques $T_i\cup \{u_i\}$.
We call the cliques $T_1$, $T_2$, $T_3$, the \emph{$T$-cliques} of the triple $\{u_1,u_2,u_3\}$. 

Let $t\in T_i$ for some $i\in \{1,2,3\}$. For $i\in \{0,1,2\}$ we say that $t$ is \emph{$i$-clique-adjacent} if $t$ has a neighbour in zero, one or two cliques of $\{T_1,T_2,T_3\}\setminus T_i$, respectively. 
By the definition of an injective colouring and the fact that every $T_i$ is a clique, a $1$-clique-adjacent vertex of $T_1\cup T_2\cup T_3$ belongs to a colour class of size at most~$2$, and a $2$-clique-adjacent vertex of $T_1\cup T_2\cup T_3$ belongs to a colour class of size~$1$. Hence, all the vertices that belong to a colour class of size~$3$ are $0$-clique-adjacent. The partition of $V(G)$ is illustrated in Figure~\ref{f-4p1}.

\begin{figure}[h]
\centering
\includegraphics[width=0.95\textwidth]{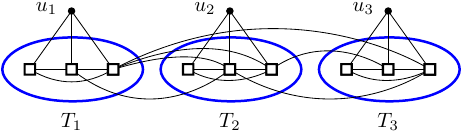}
\caption{The partition of $V(G)$ from Lemma~\ref{l-4p1}. The squares inside each $T_i$, $i \in \{1,2,3\}$, represent the sets of $0$-clique-adjacent, $1$-clique-adjacent and $2$-clique-adjacent vertices in $T_i$, respectively.}\label{f-4p1}
\end{figure}

We now use the fact that $c$ is class-$3$-optimal.
Let $t\in V\setminus \{u_1,u_2,u_3\}$, say $t\in T_1$, be $i$-clique-adjacent for $i=0$ or $i=1$. Then we may assume without loss of generality that $t$ has no neighbours in $T_2$.
If $t$ belongs to a colour class of size~$1$, then we can set $c(u_2):=c(t)$ to obtain an optimal injective colouring with fewer size-$3$ colour classes, contradicting our choice of $c$. 

We now consider the $0$-clique-adjacent vertices again. Recall that these are the only vertices, other than $u_1$, $u_2$ and $u_3$, that may belong to a colour class of size~$3$. As every $T_i$ is a clique, every colour class of size~$3$ (other than $\{u_1,u_2,u_3\}$) has exactly one vertex of each $T_i$. Let $\{w_1,w_2,w_3\}$ be another colour class of size~$3$ with $w_i\in T_i$ for every $i\in \{1,2,3\}$. Let $x\in T_1\setminus \{w_1\}$ be another $0$-clique-adjacent vertex. 
Then swapping the colours of $w_1$ and $x$ yields another class-$3$-optimal injective colouring of~$G$. Hence, we derived the following claim, which summarizes the discussion above and where statement (iv) follows from (i)--(iii).

\medskip
\noindent
{\bf Claim.} {\it Let $c$ be a class-$3$-optimal injective colouring of $G$ with $c(u_1)=c(u_2)=c(u_3)$ for three distinct vertices $u_1,u_2,u_3$ and with $p\geq 0$ other colour classes of size~$3$. Then the following four statements hold:
\begin{itemize}
\item [(i)] All $0$-clique-adjacent and $1$-clique-adjacent vertices belong to a colour class of size at least~$2$.\\[-5pt]
\item [(ii)] Let $S=\{y_1,\ldots,y_s\}$ be the set of $2$-clique-adjacent vertices. Then $\{y_1\},\ldots, \{y_s\}$ are exactly the size-$1$ colour classes.\\[-5pt]
\item [(iii)] For $i\in \{1,2,3\}$, let $x_1^i,\ldots,x_{q_i}^i$ be the $0$-clique-adjacent vertices of $T_i$ and assume without loss of generality that $q_1\leq q_2\leq q_3$. Then $p\leq q_1$ and if $p\geq 1$, we may assume without loss of generality that the size-$3$ classes, other than $\{u_1,u_2,u_3\}$, are $\{x_1^1,x_1^2,x_1^3\},\ldots, \{x_p^1,x_p^2,x_p^3\}$.\\ 
\item [(iv)] The number of colours used by $c$, or equivalently, the number of colour classes of $c$ is equal to $1+s+p+\frac{1}{2}(n-s-3(p+1))=\frac{1}{2}n+\frac{1}{2}s-\frac{1}{2}p-\frac{1}{2}$.
\end{itemize}}

\medskip
\noindent
We are now ready to present our algorithm. We first find, in polynomial time, an optimal $2$-injective colouring of $G$ by Lemma~\ref{l-easy}. We remember the number of colours used. Recall that the colour classes of every injective colouring of $G$ have size at most~$3$. So, it remains to compute an optimal injective colouring for which at least one colour class has size~$3$.

We consider each triple $u_1,u_2,u_3$ of vertices of $G$ and check if $\{u_1,u_2,u_3\}$ can be a colour class. That is, we check if $\{u_1,u_2,u_3\}$ is an independent set and has corresponding $T$-cliques $T_1$, $T_2$, $T_3$. This takes polynomial time. If not, then we discard $\{u_1,u_2,u_3\}$. Otherwise we continue as follows. Let $S=\{y_1,\ldots,y_s\}$ be the set of $2$-clique adjacent vertices in $T_1\cup T_2\cup T_3$. Exactly the vertices of $S$ will form the size-$1$ colour classes by Claim~(ii). For $i\in \{1,2,3\}$, let $x_1^i,\ldots,x_{q_i}^i$ be the $0$-clique-adjacent vertices of $T_i$, where we assume without loss of generality that $q_1\leq q_2\leq q_3$. By Claim~(iii), any injective colouring of $G$ which has $\{u_1,u_2,u_3\}$ as one of its colour classes has at most $q_1$ other colour classes of size~$3$ besides $\{u_1,u_2,u_3\}$. As can be seen from Claim~(iv), the value $\frac{1}{2}n+\frac{1}{2}s-\frac{1}{2}p-\frac{1}{2}$ is minimized if the number $p$ of size-$3$ colour classes is maximum.

From the above we can now do as follows. For $p=q_1,\ldots,1$, we check if $G$ has an injective colouring with exactly $p$ colour classes of size~$3$. We stop as soon as we find a yes-answer or if $p$ is set to $0$. We first set $\{x_1^1,x_1^2,x_1^3\},\ldots, \{x_p^1,x_p^2,x_p^3\}$ as the colour classes of size~$3$ by Claim~(iii). Let $Z$ be the set of remaining $0$-clique-adjacent and $1$-clique-adjacent vertices. We use Lemma~\ref{l-easy} to check in polynomial time if the subgraph of $G$ induced by $S\cup Z$ has an injective colouring that uses $s+\frac{1}{2}(n-s-3(p+1))$ colours (which is the minimum number of colours possible). If so, then we stop and note that after adding the size-$3$ colour classes we obtained an injective colouring of~$G$ that uses $\frac{1}{2}n+\frac{1}{2}s-\frac{1}{2}p-\frac{1}{2}$ colours, which we remember. Otherwise we repeat this step after first setting $p:=p-1$.

As the above procedure for a triple $u_1,u_2,u_3$ takes polynomial time and the number of triples we must check is $O(n^3)$, our algorithm runs in polynomial time. We take the $3$-injective colouring that uses the smallest number of colours and compare it with the number of colours used by the optimal $2$-injective colouring that we computed at the start. Our algorithm then returns a colouring with the smallest of these two values as its output. 
\end{proof}

\noindent
We use the above lemma for proving our next lemma.

\begin{lemma}\label{l-2p1p3}
{\sc Injective Colouring} is polynomial-time solvable for $(2P_1+P_3)$-free graphs.
\end{lemma}

\begin{proof}
Let $G=(V,E)$ be a $(2P_1+P_3)$-free graph. We may assume without loss of generality that $G$ is connected and by Lemma~\ref{l-4p1} that $G$ has an induced $4P_1$.
We first show that any colour class in any injective colouring of $G$ has size at most $2$. For contradiction, assume that $c$ is an injective colouring of $G$ such that there exists some colour, say colour~$1$, that has a colour class of size at least~$3$. Let $U=\{u_1,\ldots,u_p\}$ for some $p\geq 3$ be the set of vertices of $G$ with $c(u_i)=1$ for $i\in \{1,\ldots,p\}$.

As $c$ is injective, every vertex in $G-U$ has at most one neighbour in $U$. Hence, we can partition $G-U$ into (possibly empty) sets $T_0,\ldots,T_p$, where $T_0$ is the set of vertices with no neighbour in $U$ and for $i\in \{1,\ldots,p\}$, $T_i$ is the set of vertices of $G-U$ adjacent to $u_i$.

We first claim that $T_0$ is empty. For contradiction, assume $v\in T_0$. As $G$ is connected, we may assume without loss of generality that $v$ is adjacent to some vertex $t\in T_1$. Then $\{u_2,u_3,u_1,t,v\}$ induces a $2P_1+P_3$, a contradiction. Hence, $T_0=\emptyset$.

We now prove that every $T_i$ is a clique. For contradiction, assume that $t$ and $t'$ are non-adjacent vertices of $T_1$.
Then $\{u_2,u_3.t,u_1,t'\}$ induces a $2P_1+P_3$, a contradiction. Hence, every $T_i$ and thus every $T_i\cup \{u_i\}$ is a clique.

We now claim that $p=3$. For contradiction, assume that $p\geq 4$. As $G$ is connected and $U$ is an independent set, we may assume without of generality that there exist vertices $t_1\in T_1$ and $t_2\in T_2$ with $t_1t_2\in E$. Then $\{u_3,u_4,u_1,t_1,t_2\}$ induces a $2P_1+P_3$, a contradiction. Hence, $p=3$.

Now we know that $V$ can be partitioned into three cliques $T_1\cup \{u_1\}$, $T_2\cup \{u_2\}$ and $T_3\cup \{u_3\}$. However, then $G$ is $4P_1$-free, a contradiction. We conclude that every colour class of every injective colouring of $G$ has size at most~$2$. This means we can use Lemma~\ref{l-easy}.
\end{proof}

\noindent
We also use Lemma~\ref{l-4p1} in the proof of our next result.

\begin{lemma}\label{l-3p1p2}
{\sc Injective Colouring} is polynomial-time solvable for $(3P_1+P_2)$-free graphs.
\end{lemma}

\begin{proof}
Let $G$ be a $(3P_1+P_2)$-free graph on $n$ vertices. We may assume without loss of generality that $G$ is connected and by Lemma~\ref{l-4p1} that $G$ has an induced~$4P_1$. As before, we will first analyze the structure of injective colourings of $G$. We will then exploit the properties found algorithmically.

Let $c$ be an injective colouring of $G$ that has a colour class $U$ of size at least~$3$. So let $U=\{u_1,\ldots,u_p\}$ for some $p\geq 3$ be the set of vertices of $G$ with, say colour~$1$.
As $c$ is injective, every vertex in $G-U$ has at most one neighbour in $U$. Hence, we can partition $G-U$ into (possibly empty) sets $T_0,\ldots,T_p$, where $T_0$ is the set of vertices with no neighbour in $U$ and for $i\in \{1,\ldots,p\}$, $T_i$ is the set of vertices of $G-U$ adjacent to $u_i$. 

Assume that $p\geq 4$. As $G$ is connected, there exists a vertex $v\notin U$ but that has a neighbour in $U$, say $v\in T_1$. Then $\{u_2,u_3,u_4,u_1,v\}$ induces a $3P_1+P_2$, a contradiction. Hence, we have shown the following claim.

\medskip
\noindent
\emph{Claim 1. Every injective colouring of $G$ is $\ell$-injective for some $\ell\in \{1,2,3\}$.}

\medskip
\noindent
We continue as follows. As $p=3$ by Claim~1, we have $V(G)=U\cup T_0\cup T_1\cup T_2 \cup T_3$.
Suppose $T_0$ contains two adjacent vertices $x$ and $y$. Then $\{u_1,u_2,u_3,x,y\}$ induces a $3P_1+P_2$, a contradiction. Hence, $T_0$ is an independent set. As $G$ is connected, this means each vertex in $T_0$ has a neighbour in $T_1\cup T_2\cup T_3$.

Suppose $T_0$ contains two vertices $x$ and $y$ with the same colour, say $c(x)=c(y)=2$. Let $v\in T_1\cup T_2\cup T_3$, say $v\in T_1$ be a neighbour of $x$. Then, as $c(x)=c(y)$ and $c$ is injective, $v$ is not adjacent to $y$. As $T_0$ is independent, $x$ and $y$ are not adjacent. However, now $\{u_2,u_3,y,x,v\}$ induces a $3P_1+P_2$, a contradiction. Hence, every vertex in $T_0$ has a unique colour.
Suppose $T_0$ contains a vertex~$x$ and $T_1\cup T_2\cup T_3$ contains a vertex $v$ such that $c(x)=c(v)$. We may assume without loss of generality that $v\in T_1$. Then $\{u_2,u_3,x,v,u_1\}$ induces a $3P_1+P_2$, a contradiction.

Finally, suppose that $T_1\cup T_2\cup T_3$ contain two distinct vertices $v$ and $v'$ with $c(v)=c(v')$. Let $x\in T_0$. Then $x$ is not adjacent to at least one of $v$, $v'$, say $xv\notin E$ and also assume that $v\in T_1$. Then $\{u_2,u_3,x,v,u_1\}$ induces a $3P_1+P_2$. Hence, we have shown the following claim.

\medskip
\noindent
\emph{Claim 2. If $c$ is $3$-injective and $U$ is a size-$3$ colour class such that $G$ has a vertex not adjacent to any vertex of $U$, then all colour classes not equal to $U$ have size~$1$.}

\medskip
\noindent
We note that the injective colouring $c$ in Claim~2 uses $n-2$ distinct colours.

\medskip
\noindent
We continue as follows. From now on we assume that $T_0=\emptyset$. Every $T_i$ is $(P_1+P_2)$-free, as otherwise, if say $T_1$ contains an induced $P_1+P_2$, then this $P_1+P_2$, together with $u_2$ and $u_3$, forms an induced $3P_1+P_2$, which is not possible. Hence, each $T_i$ induces a complete $r_i$-partite graph for some integer $r_i$ (that is, the complement of a disjoint union of $r_i$ complete graphs). Hence, we can partition each $T_i$ into $r_i$ independent sets $T_i^1,\dots,T_i^{r_i}$ such that there exists an edge between every vertex in $T_i^a$ and every vertex in $T_i^b$ if $a\neq b$. See also Figure~\ref{f-3p1_plus_p2}.

Suppose $G$ contains another colour class of size~$3$, say $v_1$, $v_2$ and $v_3$ are three distinct vertices coloured~$2$. If two of these vertices, say $v_1$ and $v_2$, belong to the same $T_i$, say $T_1$, then $u_1$ has two neighbours with the same colour. This is not possible, as $c$ is injective. Hence, we may assume without loss of generality that $v_i\in T_i^1$ for $i\in \{1,2,3\}$.

Suppose that $T_1^2$ contains two vertices $s$ and $t$. Then, as $s$ and $t$ are adjacent to $v_1$, both of them are not adjacent to $v_2$ (recall that $c(v_1)=c(v_2)$ and $c$ is injective). Hence, $\{s,t,u_3,v_2,u_2\}$ induces a $3P_1+P_2$ (see Figure~\ref{f-3p1_plus_p2}).
We conclude that for every $i\in \{1,2,3\}$, the sets $T_i^2,\ldots,T_i^{r_i}$ have size~$1$.

\begin{figure}[h]
\centering
\includegraphics[width=0.8\textwidth]{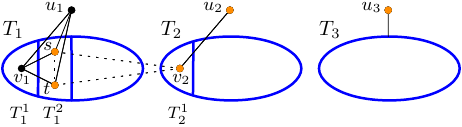}
\caption{The situation in Lemma~\ref{l-3p1p2} where $T_1^2$ contains two vertices $s$ and $t$. We show that this situation cannot happen, as it would lead to a forbidden induced $3P_1+P_2$. Note that each $u_i$ is adjacent to all vertices of $T_i$ and not to any vertices of $T_j$ for $j\neq i$. There may exist edges between vertices of different sets, but these are not drawn.}
\label{f-3p1_plus_p2}
\end{figure}

We will now make use of the fact that $G$ contains an induced $4P_1$. We note that each $T_i\cup \{u_i\}$ is a clique, unless $|T_i^1|\geq 2$. As $V(G)=T_1\cup T_2\cup T_3\cup \{u_1,u_2,u_3\}$ and $G$ contains an induced $4P_1$, we may assume without loss of generality that $T_1^1$ has size at least~$2$. Recall that $v_1\in T_1^1$. Let $z\neq v_1$ be some further vertex of~$T_1^1$.
If $z$ is not adjacent to $v_2$, then $\{z,v_1,u_3,v_2,u_2\}$ induces a $3P_1+P_2$, which is not possible. Hence, $z$ is adjacent to $v_2$. For the same reason, $z$ is adjacent to~$v_3$. This is not possible, as $c$ is injective and $v_2$ and $v_3$ both have colour~$2$. Hence, we have proven the following claim.

\medskip
\noindent
\emph{Claim 3. If $c$ is $3$-injective and $U$ is a size-$3$ colour class such that each vertex of $G-U$ is adjacent to a vertex of $U$, then $c$ has no other colour class of size~$3$.}

\medskip
\noindent
We are now ready to present our polynomial-time algorithm.
We first use Lemma~\ref{l-easy} to find in polynomial time an optimal $2$-injective colouring of $G$. We remember the number of colours it uses.

By Claim~1, it remains to find an optimal $3$-injective colouring with at least one colour class of size~$3$.
We now consider each set $\{u_1,u_2,u_3\}$ of three vertices. We discard our choice if $u_1,u_2,u_3$ do not form an independent set or if $V(G)\setminus \{u_1,u_2,u_3\}$ cannot be partitioned into sets $T_0,\ldots, T_4$ as described above.
Suppose we have not discarded our choice of vertices $u_1$, $u_2$, $u_3$. We continue as follows.

If $T_0\neq \emptyset$, then by Claim~2 the only $3$-injective colouring of $G$ (subject to colour permutation) with colour class $\{u_1,u_2,u_3\}$ is the colouring that gives each $u_i$ the same colour and a unique colour to all the other vertices of $G$. This colouring uses $n-2$ colours and we remember this number of colours.

Now suppose $T_0=\emptyset$. By Claim~3, we find that $\{u_1,u_2,u_3\}$ is the only colour class of size~$3$. Recall that no vertex in $G-\{u_1,u_2,u_3\}=T_1\cup T_2\cup T_3$ is adjacent to more than one vertex of $\{u_1,u_2,u_3\}$. Hence, we can apply Lemma~\ref{l-easy} on $G-\{u_1,u_2,u_3\}$. This yields an optimal $2$-injective colouring of $G-\{u_1,u_2,u_3\}$. We colour $u_1$, $u_2$, $u_3$ with the same colour and choose a colour that is not used in the colouring of $G-\{u_1,u_2,u_3\}$. This yields a $3$-injective colouring of $G$ that is optimal over all $3$-injective colourings with colour class $\{u_1,u_2,u_3\}$. We remember the number of colours.

As the above procedure takes polynomial time and there are $O(n^3)$ triples to consider, we find in polynomial time an optimal $3$-injective colouring of $G$ that has at least one colour class of size~$3$ (should it exist). We compare the number of colours used with the number of colours of the optimal $2$-injective colouring of~$G$ that we found earlier. Our algorithm returns the minimum of the two values as the output. Since both colourings are found in polynomial time, we conclude that our algorithm runs in polynomial time.
\end{proof}

\noindent
For proving our new hardness result we first need to introduce some terminology and prove a lemma on \textsc{Colouring}.
A $k$-colouring of $G$ can be seen as a partition of $V(G)$ into $k$ independent sets. Hence, a ($k$-)colouring of $G$ corresponds to a \emph{($k$-)clique-covering} of $\overline{G}$, which is a partition of $V(\overline{G})=V(G)$ into $k$ cliques. The \emph{clique covering number} $\overline{\chi}(G)$ of~$G$ is the smallest number of cliques in a clique-covering of $G$. Note that $\overline{\chi}(G)=\chi(\overline{G})$.

\begin{lemma}\label{l-4col}
{\sc Colouring} is \NP-complete for graphs with $\overline{\chi}\leq 3$.
\end{lemma}

\begin{proof}
The \textsc{List Colouring} problem takes as input a graph $G$ and a \emph{list assignment} $L$ that assigns each vertex $u\in V(G)$ a list $L(u)\subseteq \{1,2,\ldots\}$.
The question is whether $G$ admits a colouring $c$ with $c(u)\in L(u)$ for every $u\in V(G)$.
Jansen~\cite{Ja96} proved that \textsc{List Colouring} is \NP-complete for co-bipartite graphs. This is the problem we reduce from.

Let $G$ be a graph with a list assignment $L$ and assume that $V(G)$ can be split into two (not necessarily disjoint) cliques $K$ and $K'$. We set $A_1:=K$ and $A_2:=K\setminus K'$.
As both $A_1$ and $A_2$ are cliques, we have that $\overline{\chi}(G)\leq 2$.
We may assume without loss of generality that the union of all the lists $L(u)$ is $\{1,\ldots,k\}$ for some integer $k$.
We now extend $G$ by adding a clique $A_3$ of $k$ new vertices $v_1,\ldots,v_k$ and by adding an edge between a vertex $v_\ell$ and a vertex~$u\in V(G)$ if and only if $\ell\notin L(u)$.
This yields a new graph $G'$ with $\overline{\chi}(G')\leq 3$. It is readily seen that $G$ has a colouring $c$ with $c(u)\in L(u)$ for every $u\in V(G)$ if and only if $G'$ has a $k$-colouring.
\end{proof}

\noindent
We use Lemma~\ref{l-4col} to prove the next lemma.

\begin{lemma}\label{l-5p1}
{\sc Injective Colouring} is \NP-complete for $5P_1$-free graphs.
\end{lemma}

\begin{proof}
The problem is readily seen to belong to \NP.
We reduce from \textsc{Colouring}. Let $(G,k)$ be an instance of this problem. By Lemma~\ref{l-4col} we may assume that $V(G)$ can be partitioned into three cliques $A_1$, $A_2$ and $A_3$ with $|A_1|\leq |A_2|\leq |A_3|$.
We may assume that $k\geq |A_3|$; otherwise $(G,k)$ is a no-instance. Moreover, we may assume that every vertex $u$ in every~$A_i$ has at least one neighbour in $V\setminus A_i$; otherwise $u$ has degree $|A_i|-1\leq k-1$ and hence, $G-u$ is $k$-colourable if and only if $G$ is $k$-colourable.

We now construct a graph $G'$ as follows. Let $E^*$ be the set of edges in $G$ whose end-vertices belong to different cliques of $\{A_1,A_2,A_3\}$.
We add a clique $A_0$ of $|E^*|$ new vertices, so with exactly one vertex $v_e$ for each edge $e=xy$ in $E^*$.
We replace each $e\in E^*$ by the edges $xv_e$ and $yv_e$. We denote the resulting graph by~$G'$ (see also Figure~\ref{f-5p1}).
We claim that $G$ has a $k$-colouring if and only if $G'$ has an injective $(k+|E^*|)$-colouring.

\begin{figure}[h]
\centering
\includegraphics[width=0.35\textwidth]{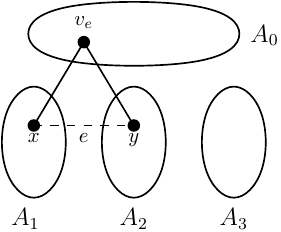}
\caption{The graph $G'$ constructed in the proof of Lemma~\ref{l-5p1}.}
\label{f-5p1}
\end{figure}

First suppose that $G$ has a $k$-colouring $c$. We give each vertex of $A_0$ a unique colour from $\{k+1,\ldots,k+|E^*|\}$. This yields a $(k+|E^*|)$-colouring $c'$ of $G'$. We claim that $c'$ is injective. In order to see this, suppose that $G'$ contains a vertex $s$ that has two neighbours $x$ and $y$ with $c'(x)=c'(y)$. Every vertex in $A_1\cup A_2 \cup A_3$ is only adjacent to vertices from its own clique $A_i$ and $A_0$ and the colour sets used on those two cliques do not intersect.
Hence, $s$ belongs to $A_0$. Then, by definition of $G'$, we find that $x$ and $y$ must belong to different cliques $A_h$ and $A_i$. By construction, $xy$ is an edge in $E$. As $c$ is a $k$-colouring, this means that $c'(x)=c(x)\neq c(y)=c'(y)$, a contradiction. We conclude that $c'$ is an injective $(k+|E^*|)$-colouring of $G'$.

Now suppose that $G'$ has a $(k+|E^*|)$-colouring $c'$. Let $e\in A_0$ and suppose $c'(e)=1$. We assume without loss of generality that $e$ corresponds to an edge $e=xy$ in $G$ with $x\in A_1$ and $y\in A_2$. Then, in $G'$, we have that $e$ is adjacent to $x$ and to $y$. Hence, $x$ and $y$ are not coloured~$1$. As $c'$ is injective, the neighbours of $x$ and $y$ have different colours.
As $A_1$ and $A_2$ are cliques, $x$ is adjacent to every vertex in $A_1\setminus \{x\}$ and $y$ is adjacent to every vertex in $A_2\setminus \{y\}$.
Hence, no vertex in $A_1\cup A_2$ can have colour~$1$.

Now suppose that there exists a vertex $z\in A_3$ with $c'(z)=1$. In $G$ each vertex in every $A_i$ has at least one neighbour in a different clique $A_j$. Hence, $z$ has a neighbour $f\in A_0$ in $G'$ by construction of $G'$. However, now $f$ has two neighbours, $e$ and $z$, each with colour~$1$, contradicting the fact that $c'$ is injective.
We conclude that the colours of $A_0$ do not occur on $A_1\cup A_2\cup A_3$.

Recall that $A_0$ is a clique of size $|E^*|$. Hence, $c'$ uses $|E^*|$ different colours. As no colour of $A_0$ occurs on $A_1\cup A_2\cup A_3$, this means that $|E^*|$ colours are not used on $V(G)$. Hence, the restriction $c$ of $c'$ to $V(G)=A_1\cup A_2 \cup A_3$ is a $k$-colouring of the subgraph of $G'$ induced by $A_1\cup A_2\cup A_3$.

We claim that $c$ is even a $k$-colouring of $G$. Otherwise, if there exists an edge $e=xy$ with $c(x)=c(y)$, then $e$ must be an edge in $G$ that is not in $G'$.
This means that $x$ and $y$ must belong to different cliques $A_i$ and $A_j$. By construction, $G'$ then contains the vertex $e=xy$.
However, then $c'(x)=c(x)=c(y)=c(y')$ and $e'$ has two neighbours with the same colour. This contradicts our assumption that $c'$ is injective. We conclude that $c$ is a $k$-colouring of~$G$.
\end{proof}

\noindent
We combine the above results with results of Bodlaender et al.~\cite{BKTL04} and Mahdian~\cite{Ma02} to prove Theorem~\ref{t-injective}.

\medskip
\noindent
{\bf Theorem~\ref{t-injective} (restated).}
\emph{Let $H$ be a graph. For the class of $H$-free graphs it holds that:\\[-10pt]
\begin{enumerate}
\item [(i)] {\sc Injective Colouring} is polynomial-time solvable if $H\ssin 2P_1+P_4$ and \NP-complete if $H\not\ssi 2P_1+P_4$;\\[-10pt]
\item [(ii)] for every $k\geq 4$, {\sc Injective $k$-Colouring} is polynomial-time solvable if $H$ is a linear forest and \NP-complete otherwise.
\end{enumerate}}

\begin{proof}
We first prove (ii). If $C_3\ssi H$, then we use Lemma~\ref{l-triangle}.
Now suppose $C_p\ssi H$ for some $p\geq 4$. Mahdian~\cite{Ma02} proved that for every $g\geq 4$ and $k\geq 4$, \textsc{Injective $k$-Colouring} is \NP-complete for line graphs of bipartite graphs of girth at least~$g$.
These graphs may not be $C_3$-free but are $C_p$-free for $g\geq p+1$. Now assume $H$ has no cycle, so $H$ is a forest.
If $H$ contains a vertex of degree at least~$3$, then $H$ contains an induced $K_{1,3}$.
As every line graph is $K_{1,3}$-free, we can use the aforementioned result of Mahdian~\cite{Ma02} again.
Otherwise $H$ is a linear forest, in which case we use Corollary~\ref{c-linearforest}.

We now prove (i). Due to (ii), we may assume that $H$ is a linear forest. If $H\ssi P_1+P_4$ or $H\ssi 2P_1+P_3$ or $H\ssi 3P_1+P_2$, then we use Lemma~\ref{l-p1p4},~\ref{l-2p1p3}, or~\ref{l-3p1p2}, respectively.
Hence, if $H\ssin 2P_1+P_4$, then \textsc{Injective Colouring} is polynomial-time solvable for $H$-free graphs. Now suppose that $H\not\ssi 2P_1+P_4$.
If $2P_2\ssi H$, then the class of $(2P_2,C_4,C_5)$-free graphs are contained in the class of $H$-free graphs. The latter class coincides with the class of split graphs~\cite{FH77}.
Recall that Bodlaender et al.~\cite{BKTL04} proved that \textsc{Injective Colouring} is \NP-complete for split graphs. In the remaining case it holds that $5P_1\ssi H$, and for this case we can use Lemma~\ref{l-5p1}.
\end{proof}

\section{Conclusions}\label{s-con}

Our complexity study led to three complete and three almost complete complexity classifications (Theorems~\ref{t-acyclic}--\ref{t-injective}). Due to our systematic approach we were able to identify a number of open questions for future research, which we collect below.

In Lemma~\ref{l-girth} we prove that for every $k\geq 3$ and every $g\geq 3$, \textsc{Acyclic $k$-Colouring} is \NP-complete for graphs of girth at least~$g$.
We would like to prove an analogous result for the third problem we considered.
We recall that \textsc{Injective $3$-Colouring} is polynomial-time solvable for general graphs.
Moreover, for every $k\geq 4$, \textsc{Injective $k$-Colouring} is \NP-complete for bipartite graphs (by Lemma~\ref{l-triangle}) and thus for graphs of girth at least~$4$.
Hence, we pose the following open problem.

\begin{open}\label{o-othertwo}
For every $g\geq 5$, determine the complexity of {\sc Injective Colouring} and {\sc Injective $k$-Colouring} ($k\geq 4$) for graphs of girth at least $g$.
\end{open}

\noindent
This problem has eluded us and remains open and is, we believe, challenging.
We have made progress for the corresponding high-girth problem for \textsc{Star $3$-Colouring} in Lemma~\ref{lem:star-col-high-girth}.
However, we leave the high-girth problem for \textsc{Star $k$-Colouring} open for $k\geq 4$, as follows.
We believe it represents an interesting technical challenge. At the moment, we only know that for $k\geq 4$,
\textsc{Star $k$-Colouring} is \NP-complete for bipartite graphs~\cite{ACKKR04} and thus for graphs of girth at least~$4$.

\begin{open}\label{o-star}
For every $g\geq 5$, determine the complexity of {\sc Star $k$-Colouring} ($k\geq 4$) for graphs of girth at least $g$.
\end{open}

\noindent
Naturally we also aim to settle the remaining open cases for our three problems in Table~\ref{t-thetable}.
In particular, there is one case left for each of the problems {\sc Acyclic Colouring}, \textsc{Star Colouring}, and \textsc{Injective Colouring}.
We note that the graph $G'$ in the proof of Lemma~\ref{l-5p1} contains an induced $2P_1+P_4$.

\begin{open}\label{o-5left}
Determine the complexity of {\sc Injective Colouring} for $(2P_1+P_4)$-free graphs.
\end{open}

\begin{open}\label{o-2p2}
Determine the complexity of {\sc Acyclic Colouring} and {\sc Star Colouring} for $2P_2$-free graphs.
\end{open}

\noindent
Recall that \textsc{Injective Colouring} and \textsc{Colouring} are \NP-complete for $2P_2$-free graphs.
However, none of the hardness constructions for these problems carry over to \textsc{Acyclic Colouring} and \textsc{Star Colouring}.
In this context, the next open problem from Lyons~\cite{Ly09} for a subclass of $2P_2$-free graphs is also interesting.
A graph $G=(V,E)$ is \emph{split} if $V=I\cup K$, where $I$ is an independent set, $K$ is a clique and $I\cap K=\emptyset$.
The class of split graphs coincides with the class of $(2P_2,C_4,C_5)$-free graphs~\cite{FH77} and thus \textsc{Acyclic Colouring} is equivalent to \textsc{Colouring} for split graphs, and hence it is polynomial-time solvable.
However, for \textsc{Star Colouring} this equivalence is no longer true.

\begin{open}[\cite{Ly09}]\label{o-split}
Determine the complexity of {\sc Star Colouring} for split graphs, or equivalently, $(2P_2,C_4,C_5)$-free graphs.
\end{open}

\noindent
Let $\omega(G)$ denote the clique number of $G$ (size of a largest clique of $G$).
Let $\chi_s(G)$ denote the the star chromatic number of $G$.
It is easily observed (see also~\cite{Ly09}) that if $G$ is a split graph, then either $\chi_s(G)=\omega(G)$ or $\chi_s(G)=\omega(G)+1$.

\medskip
\noindent
Finally, we recall that \textsc{Injective Colouring} is also known as $L(1,1)$-labelling.
The general distance constrained labelling problem \textsc{$L(a_1,\ldots,a_p)$-Labelling} is to decide if a graph $G$ has a labelling $c:V(G)\to \{1,\ldots,k\}$ for some integer $k\geq 1$ such that for every $i\in \{1,\ldots,p\}$, $|c(u)-c(v)|\geq a_i$ whenever $u$ and $v$ are two vertices of distance~$i$ in $G$ (in this setting, it is usually assumed that $a_1\geq \ldots \geq a_p$). If $k$ is fixed, we write \textsc{$L(a_1,\ldots,a_p)$-$k$-Labelling} instead. By applying Theorem~\ref{t-general} we obtain the following result.

\begin{theorem}\label{l-la1ap}
For all $k\geq 1, a_1\geq 1,\ldots, a_k\geq 1$, the {\sc $L(a_1,\ldots,a_p)$-$k$-Labelling} problem is polynomial-time solvable for $H$-free graphs if $H$ is a linear forest.
\end{theorem}

\noindent
We leave a more detailed and systematic complexity study of problems in this framework for future work (see, for example,~\cite{Ca11,FGKLP12,FKK01} for some complexity results for both general graphs and special graph~classes).

\medskip
\noindent
\emph{Acknowledgments.} We thank Cyriac Antony for some helpful comments.


\begin{thebibliography}{10}

\bibitem{ACKKR04}
Michael~O. Albertson, Glenn~G. Chappell, Henry~A. Kierstead, Andr{\'{e}}
 K{\"{u}}ndgen, and Radhika Ramamurthi.
\newblock Coloring with no 2-colored ${P}_4$'s.
\newblock \emph{Electronic Journal of Combinatorics}, 11, 2004.

\bibitem{ACGMRSWY20}
Noga Alon, Jonathan D. Cohen, Thomas L. Griffiths, Pasin Manurangsi, Daniel Reichman, Igor Shinkar, Tal Wagner, and Alexander Yu.
\newblock Multitasking capacity: Hardness results and improved constructions.
\newblock \emph{SIAM Journal on Discrete Mathematics}, 34:885-903, 2020.

\bibitem{AMR91}
Noga Alon, Colin McDiarmid, and Bruce~A. Reed.
\newblock Acyclic coloring of graphs.
\newblock \emph{Random Structures and Algorithms}, 2:277--288, 1991.

\bibitem{AZ02}
Noga Alon and Ayal Zaks.
\newblock Algorithmic aspects of acyclic edge colorings.
\newblock \emph{Algorithmica}, 32:611--614, 2002.

\bibitem{AF12}
Patrizio Angelini and Fabrizio Frati.
\newblock Acyclically 3-colorable planar graphs.
\newblock \emph{Journal of Combinatorial Optimization}, 24:116--130, 2012.

\bibitem{ALR12}
Aistis Atminas, Vadim~V. Lozin, and Igor Razgon.
\newblock Linear time algorithm for computing a small biclique in graphs without long induced paths.
\newblock \emph{Proceedings of {S}{W}{A}{T} 2012, LNCS}, 7357:142--152, 2012.

\bibitem{Bo96}
Hans~L. Bodlaender.
\newblock A linear-time algorithm for finding tree-decompositions of small treewidth.
\newblock \emph{{SIAM} Journal on Computing}, 25:1305--1317, 1996.

\bibitem{BKTL04}
Hans~L. Bodlaender, Ton Kloks, Richard~B. Tan, and Jan van Leeuwen.
\newblock Approximations for lambda-colorings of graphs.
\newblock \emph{Computer Journal}, 47:193--204, 2004.

\bibitem{BJMPS20}
Jan Bok, Nikola Jedli{}\u{c}kov\'{a}, Barnaby Martin, Dani\"el Paulusma and Siani Smith.
Acyclic Colouring, Star Colouring and Injective Colouring for ${H}$-Free Graphs.
\emph{Proc. ESA 2020, LIPIcs} 173, 22:1-22:22, 2020.

\bibitem{BJMPS21}
Jan Bok, Nikola Jedli{}\u{c}kov\'{a}, Barnaby Martin, Dani\"el Paulusma and Siani Smith. Injective colouring for $H$-free graphs. \emph{Proc. CSR 2021, LNCS}, to appear.

\bibitem{BDFJP19}
Marthe Bonamy, Konrad~K. Dabrowski, Carl Feghali, Matthew Johnson, and Dani\"el Paulusma.
\newblock Independent feedback vertex set for ${P}_5$-free graphs.
\newblock \emph{Algorithmica}, 81:1342--1369, 2019.

\bibitem{Bo79}
Oleg~V. Borodin.
\newblock On acyclic colorings of planar graphs.
\newblock \emph{Discrete Mathematics}, 25:211--236, 1979.

\bibitem{BGMPS}
C. Brause, P.A. Golovach, B. Martin, D. Paulusma and S. Smith. Acyclic, star and injective colouring: bounding the diameter. 
\newblock \emph{CoRR}, abs/2104.10593, 2021.

\bibitem{BGPS12}
Hajo Broersma, Petr~A. Golovach, Dani{\"{e}}l Paulusma, and Jian Song.
\newblock Updating the complexity status of coloring graphs without a fixed induced linear forest.
\newblock \emph{Theoretical Computer Science}, 414:9--19, 2012.

\bibitem{Ca11}
Tiziana Calamoneri.
\newblock The ${L}(h,k)$-labelling problem: An updated survey and annotated bibliography.
\newblock \emph{Computer Journal}, 54:1344--1371, 2011.

\bibitem{CMS11}
Christine~T. Cheng, Eric McDermid, and Ichiro Suzuki.
\newblock Planarization and acyclic colorings of subcubic claw-free graphs.
\newblock \emph{Proc. of {WG} 2011, LNCS}, 6986:107--118, 2011.

\bibitem{CHSZ18}
Maria Chudnovsky, Shenwei Huang, Sophie Spirkl, and Mingxian Zhong.
\newblock List-three-coloring graphs with no induced ${P}_6+r{P}_3$.
\newblock \emph{CoRR}, abs/1806.11196, 2018.

\bibitem{CC86}
Thomas~F. Coleman and Jin-Yi Cai.
\newblock The cyclic coloring problem and estimation of sparse {H}essian matrices.
\newblock \emph{SIAM Journal on Algebraic Discrete Methods}, 7:221--235, 1986.

\bibitem{Co90}
Bruno Courcelle.
\newblock The monadic second-order logic of graphs. {I}. {R}ecognizable sets of finite graphs.
\newblock \emph{Information and Computation}, 85:12--75, 1990.

\bibitem{DMS13}
Zden\v{e}k Dvo\v{r}{\'{a}}k, Bojan Mohar, and Robert \v{S}{\'{a}}mal.
\newblock Star chromatic index.
\newblock \emph{Journal of Graph Theory}, 72(3):313--326, 2013.

\bibitem{EHK98}
Thomas Emden{-}Weinert, Stefan Hougardy, and Bernd Kreuter.
\newblock Uniquely colourable graphs and the hardness of colouring graphs of large girth.
\newblock \emph{Combinatorics, Probability and Computing}, 7:375--386, 1998.

\bibitem{Er59}
Paul Erd\H{o}s.
\newblock Graph theory and probability.
\newblock \emph{Canadian Journal of Mathematics}, 11:34--38, 1959.

\bibitem{FGR02}
Guillaume Fertin, Emmanuel Godard, and Andr{\'{e}} Raspaud.
\newblock Minimum feedback vertex set and acyclic coloring.
\newblock \emph{Information Processing Letters}, 84:131--139, 2002.

\bibitem{FH77}
S.~F\"oldes and P.~L. Hammer.
\newblock Split graphs.
\newblock \emph{Congressus Numerantium}, 19:311--315, 1977.

\bibitem{FR08}
Guillaume Fertin and Andr{\'{e}} Raspaud.
\newblock Acyclic coloring of graphs of maximum degree five: Nine colors are enough.
\newblock \emph{Information Processing Letters}, 105:65--72, 2008.

\bibitem{FRR04}
Guillaume Fertin, Andr{\'{e}} Raspaud, and Bruce~A. Reed.
\newblock Star coloring of graphs.
\newblock \emph{Journal of Graph Theory}, 47(3):163--182, 2004.

\bibitem{FGK11}
Ji\v{r}{\'{\i}} Fiala, Petr~A. Golovach, and Jan Kratochv{\'{\i}}l.
\newblock Parameterized complexity of coloring problems: Treewidth versus vertex cover.
\newblock \emph{Theoretical Computer Science}, 412:2513--2523, 2011.

\bibitem{FGKLP12}
Ji\v{r}{\'{\i}} Fiala, Petr~A. Golovach, Jan Kratochv{\'{\i}}l, Bernard Lidick{\'{y}}, and Dani{\"{e}}l Paulusma.
\newblock Distance three labelings of trees.
\newblock \emph{Discrete Applied Mathematics}, 160:764--779, 2012.

\bibitem{FKK01}
Ji\v{r}{\'{\i}} Fiala, Ton Kloks, and Jan Kratochv{\'{\i}}l.
\newblock Fixed-parameter complexity of lambda-labelings.
\newblock \emph{Discrete Applied Mathematics}, 113:59--72, 2001.

\bibitem{GLPR19}
Esther Galby, Paloma~T. Lima, Dani{\"{e}}l Paulusma, and Bernard Ries.
\newblock Classifying \emph{k}-edge colouring for \emph{H}-free graphs.
\newblock \emph{Information Processing Letters}, 146:39--43, 2019.


\bibitem{GJS76}
Michael~R. Garey, David~S. Johnson, and Larry~J. Stockmeyer.
\newblock Some simplified \NP-complete graph problems.
\newblock \emph{Theoretical Computer Science}, 1:237--267, 1976.

\bibitem{GJPS17}
Petr~A. Golovach, Matthew Johnson, Dani\"el Paulusma, and Jian Song.
\newblock A survey on the computational complexity of colouring graphs with forbidden subgraphs.
\newblock \emph{Journal of Graph Theory}, 84:331--363, 2017.

\bibitem{GPS14}
Petr~A. Golovach, Dani{\"{e}}l Paulusma, and Jian Song.
\newblock Coloring graphs without short cycles and long induced paths.
\newblock \emph{Discrete Applied Mathematics}, 167:107--120, 2014.

\bibitem{HKSS02}
Ge\v{n}a Hahn, Jan Kratochv{\'{\i}}l, Jozef \v{S}ir{\'{a}}\v{n}, and Dominique Sotteau.
\newblock On the injective chromatic number of graphs.
\newblock \emph{Discrete Mathematics}, 256:179--192, 2002.

\bibitem{HRS08}
Pavol Hell, Andr{\'{e}} Raspaud, and Juraj Stacho.
\newblock On injective colourings of chordal graphs.
\newblock \emph{Proc. {LATIN} 2008, LNCS}, 4957:520--530, 2008.

\bibitem{Ho81}
Ian Holyer.
\newblock The {N}{P}-completeness of edge-coloring.
\newblock \emph{{SIAM} Journal on Computing}, 10:718--720, 1981.

\bibitem{HJP15}
Shenwei Huang, Matthew Johnson, and Dani{\"{e}}l Paulusma.
\newblock Narrowing the complexity gap for colouring $({C}_s,{P}_t)$-free graphs.
\newblock \emph{Computer Journal}, 58:3074--3088, 2015.

\bibitem{JKM09}
Robert Janczewski, Adrian Kosowski, and Michal Malafiejski.
\newblock The complexity of the ${L}(p,q)$-labeling problem for bipartite planar graphs of small degree.
\newblock \emph{Discrete Mathematics}, 309:3270--3279, 2009.

\bibitem{Ja96}
K. Jansen.
Complexity Results for the Optimum Cost Chromatic Partition Problem.
Universit\"at Trier, Mathematik/Informatik, Forschungsbericht 96--41, 1996.

\bibitem{JXZ13}
Jing Jin, Baogang Xu, and Xiaoyan Zhang.
\newblock On the complexity of injective colorings and its generalizations.
\newblock \emph{Theoretical Computer Science}, 491:119--126, 2013.

\bibitem{KM11}
Ross~J. Kang and Tobias M{\"{u}}ller.
\newblock Frugal, acyclic and star colourings of graphs.
\newblock \emph{Discrete Applied Mathematics}, 159:1806--1814, 2011.

\bibitem{Ka18}
T.~Karthick.
\newblock Star coloring of certain graph classes.
\newblock \emph{Graphs and Combinatorics}, 34:109--128, 2018.

\bibitem{KMMNPS18}
Tereza Klimo\v{s}ov\'a, Josef Mal\'ik, Tom\'a\v{s} Masa\v{r}\'ik, Jana Novotn\'a, Dani\"el Paulusma, and Veronika Sl\'ivov\'a.
\newblock Colouring $({P}_r+{P}_s)$-free graphs.
\newblock \emph{Proc. ISAAC 2018, LIPIcs}, 123:5:1--5:13, 2018.

\bibitem{Ko78}
Alexander V. Kostochka, Upper bounds of chromatic functions of graphs. PhD thesis, University of Novosibirsk, 1978.

\bibitem{KKTW01}
Daniel Kr{\'a}l', Jan Kratochv\'{\i}l, {\relax Zs}olt Tuza, and Gerhard~J. Woeginger.
\newblock Complexity of coloring graphs without forbidden induced subgraphs.
\newblock \emph{Proc. WG 2001, LNCS}, 2204:254--262, 2001.

\bibitem{LSS18}
Hui Lei, Yongtang Shi, and Zi-Xia Song.
\newblock Star chromatic index of subcubic multigraphs.
\newblock \emph{Journal of Graph Theory}, 88:566--576, 2018.

\bibitem{LG83}
Daniel Leven and Zvi Galil.
\newblock {NP}-completeness of finding the chromatic index of regular graphs.
\newblock \emph{Journal of Algorithms}, 4:35--44, 1983.

\bibitem{SMMS14}
Cl{\'{a}}udia Linhares-Sales, Ana~Karolinna Maia, N{\'{\i}}colas~A. Martins, and Rudini~Menezes Sampaio.
\newblock Restricted coloring problems on graphs with few ${P}_4$'s.
\newblock \emph{Annals of Operations Research}, 217:385--397, 2014.

\bibitem{LR92}
Errol~L. Lloyd and Subramanian Ramanathan.
\newblock On the complexity of distance-$2$ coloring.
\newblock \emph{Proc. ICCI 1992}, pages 71--74, 1992.

\bibitem{LK07}
Vadim~V. Lozin and Marcin Kami\'nski.
\newblock Coloring edges and vertices of graphs without short or long cycles.
\newblock \emph{Contributions to Discrete Mathematics}, 2(1), 2007.

\bibitem{Ly11}
Andrew Lyons.
\newblock Acyclic and star colorings of cographs.
\newblock \emph{Discrete Applied Mathematics}, 159:1842--1850, 2011.

\bibitem{Ly09}
Andrew Lyons.
Restricted coloring problems and forbidden induced subgraphs.
Manuscript, 2009.

\bibitem{Ma02}
Mohammad Mahdian.
\newblock On the computational complexity of strong edge coloring.
\newblock \emph{Discrete Applied Mathematics}, 118:239--248, 2002.

\bibitem{MNRW13}
Debajyoti Mondal, Rahnuma~Islam Nishat, Md.~Saidur Rahman, and Sue Whitesides.
\newblock Acyclic coloring with few division vertices.
\newblock \emph{Journal of Discrete Algorithms}, 23:42--53, 2013.

\bibitem{Oc05}
Pascal Ochem, Graph Coloring and Combinatorics on Words.
PhD Thesis, Universit\'e de Bordeaux, 2005.

\bibitem{RSST97}
Neil Robertson, Daniel~P. Sanders, Paul~D. Seymour, and Robin Thomas.
\newblock The four-colour theorem.
\newblock \emph{Journal of Combinatorial Theory, Series {B}}, 70:2--44, 1997.

\bibitem{SH97}
Arunabha Sen and Mark~L. Huson.
\newblock A new model for scheduling packet radio networks.
\newblock \emph{Wireless Networks}, 3:71--82, 1997.

\bibitem{SA20}
M. A. Shalu and Cyriac Antony.
Complexity of Restricted Variant of Star Colouring.
\emph{Proc. CALDAM 2020, LNCS} 12016:3--14, 2020.

\bibitem{Vi64}
Vadim~Georgievich Vizing.
\newblock On an estimate of the chromatic class of a $p$-graph.
\newblock \emph{Diskret Analiz}, 3:25--30, 1964.

\bibitem{Wo05}
David~R. Wood.
\newblock Acyclic, star and oriented colourings of graph subdivisions.
\newblock \emph{Discrete Mathematics and Theoretical Computer Science}, 7:37--50, 2005.

\bibitem{ZB04}
Xiao{-}Dong Zhang and Stanislaw Bylka.
\newblock Disjoint triangles of a cubic line graph.
\newblock \emph{Graphs and Combinatorics}, 20:275--280, 2004.

\bibitem{ZKN00}
Xiao Zhou, Yasuaki Kanari, and Takao Nishizeki.
\newblock Generalized vertex-coloring of partial $k$-trees.
\newblock \emph{IEICE Transactions on Fundamentals of Electronics, Communication and Computer Sciences}, E83-A:671--678, 2000.

\bibitem{ZLSX16}
Enqiang Zhu, Zepeng Li, Zehui Shao, and Jin Xu.
\newblock Acyclically 4-colorable triangulations.
\newblock \emph{Information Processing Letters}, 116:401--408, 2016.

\end{thebibliography}
\end{document}